\documentclass[11 pt]{article}
\usepackage{amsmath, mathtools}
\usepackage{amsthm}
\usepackage{amssymb}
\usepackage{graphicx}
\usepackage{appendix}
\usepackage{color}

\newtheorem{definition}{Definition}[section]
\newtheorem{assumption}[definition]{Assumption}
\newtheorem{theorem}[definition]{Theorem}
\newtheorem{lemma}[definition]{Lemma}
\numberwithin{equation}{section}

\begin{document}

\title{Analysis of optimal portfolio on finite and small time horizons for a stochastic volatility market model }
\author{ Minglian Lin\footnote{Email: minglian.lin@ndsu.edu} \quad and \quad Indranil SenGupta\footnote{Email: indranil.sengupta@ndsu.edu} \\ Department of Mathematics \\ North Dakota State University}

\date{\today}

\maketitle

\begin{abstract}
In this paper, we consider the portfolio optimization problem in a financial market under a general utility function. Empirical results suggest that if a significant market fluctuation occurs, invested wealth tends to have a notable change from its current value. We consider an incomplete stochastic volatility market model, that is driven by both a Brownian motion and a jump process. At first, we obtain a closed-form formula for an approximation to the optimal portfolio in a small-time horizon. This is obtained by finding the associated Hamilton-Jacobi-Bellman integro-differential equation and then approximating the value function by constructing appropriate super-solution and sub-solution. It is shown that the true value function can be obtained by sandwiching the constructed super-solution and sub-solution. We also prove the accuracy of the approximation formulas. Finally, we provide a procedure for generating a close-to-optimal portfolio for a finite time horizon.

\end{abstract}

	
\section{Introduction}

The problem of portfolio optimization is of paramount importance in financial industries. This problem is particularly interesting if the corresponding process is highly volatile. If the wealth of a portfolio is $x$ at an initial time $t$, the usual goal is to invest in such a way that maximizes the expected utility of wealth at the time horizon $T$. Mathematically, if $U_T$ is a function modeling the investor's utility of wealth at time $T$, then the goal is typically to find a portfolio $\pi$, that maximizes $\mathbb{E}(U_T| \mathcal{F}_t)$, where $\mathcal{F}_t$ represents the $\sigma$-algebra of information up to the initial time $t$.

In the pioneering paper \cite{Merton}, Merton examines the combined problem of optimal portfolio selection and consumption rules for an individual in a continuous-time model where the investor's income is generated by stochastic returns on assets. The optimality equations for a multi-asset problem are derived when the rate of returns is generated by a standard Brownian motion. The work is further developed in Merton's following paper \cite{Merton2}, which extends the previous results for more general utility functions, price behavior assumptions, and for income generated also from noncapital gains sources.

The procedure is developed in a great number of other works in literature. For example, in the paper \cite{Sircar1}, the finite horizon Merton portfolio optimization problem is studied in a general local-stochastic volatility setting. It is shown that the zeroth-order approximation of the value function and optimal investment strategy correspond to those obtained in \cite{Merton}, when the risky asset follows a geometric Brownian motion. In addition to the finite time horizon, a similar problem is studied under the assumption of an infinite time horizon (see \cite{Pang1, Pang2}).  Usually, the markets are incomplete, which stems from the presence of a stochastic factor that affects the dynamics of the correlated stock price. In \cite{Nadtochiy}, the authors provide an approximation scheme for the maximal expected utility and optimal investment policies for the portfolio choice problem in an incomplete market. The analysis is developed on a splitting of the Hamilton-Jacobi-Bellman (HJB) equation in two sub-equations.

For an incomplete market, it is challenging to obtain an explicit formula for the optimal portfolio. Some attempts have been made in that direction. However, most of such approaches involve some strong assumptions. For example, the paper \cite{Zariphopoulou} studies a class of stochastic optimization models of expected utility in markets with randomly changing investment opportunities. Under certain assumptions on the individual preferences, an explicit solution is obtained. It is shown that with an appropriate transformation, the value function is expressible in terms of the solution of a linear parabolic equation. In \cite{Wang}  a continuous-time robust mean-variance model in the jump-diffusion financial market with an intractable claim is considered. More importantly, in that work, the price processes of the assets are driven by both the Brownian motion and Poisson jumps. Once again, under suitable assumptions, an explicit closed-form solution of the robust mean-variance portfolio selection model is obtained.

In order to reduce the number of assumptions involved in obtaining an explicit formula, ``approximation approaches" are popular in the literature. In \cite{Fouque}, the authors study the Merton portfolio optimization problem in the presence of stochastic volatility using asymptotic approximations. The authors derive formal derivations of asymptotic approximations. In addition,  a convergence proof is obtained in the case of power utility and single-factor stochastic volatility. The primary contribution of \cite{Fouque} is to be able to consider the portfolio optimization problem with general utility functions and in the context of incomplete stochastic volatility markets. This is obtained by decomposing volatility into two components - one on a fast timescale and one on a slow timescale. This is in agreement with \cite{Chacko}, which alludes that two volatility factors, one fast and one slow, need to be considered simultaneously.

The incorporation of ``jumps" in portfolio optimization has gained its importance in literature. As observed in \cite{Jin}, if a jump occurs, invested wealth tends to have a significant change from its current value. Such changes are difficult to hedge through continuous rebalancing. Consequently, this leads to potentially large losses for investors with short positions.  The paper \cite{Sahalia} analyzes the consumption-portfolio selection problem of an investor facing both Brownian and jump risks. The study provides a solution in closed form. The closed-form solution provides insights into the structure of the optimal portfolio in the presence of jumps. It is shown that for a single jump term the optimal investment policy can be summarized by an appropriate theorem. The study in \cite{Das} considers a portfolio optimization problem, where the market consisting of several stocks is modeled by a multidimensional jump-diffusion process with age-dependent semi-Markov modulated coefficients. On the finite time horizon, the authors discuss the risk-sensitive portfolio optimization using a probabilistic approach and establish the existence and uniqueness of the classical solution to the corresponding HJB equation.

In \cite{Kumar} this problem is considered in a simple incomplete market and under a general utility function. The authors obtain a closed-form formula for a trading strategy that approximates
the optimal trading strategy. The analysis implements the associated HJB partial differential equation. Also, the associated risky process is assumed to be driven by only Brownian motions. For the small-time horizon, the approximate value function is obtained by constructing classical sub-solution and super-solution to the HJB partial differential equation using a formal expansion in powers of horizon time. The method of super-solution and sub-solution to bind the solution of interest has been recently used to optimize portfolios in the commodity market (see \cite{Roberts1, Roberts2}). In those works, however, the problem is studied from the point of view of the sequential decision-making problem. In \cite{Kumar}, the authors also provide a heuristic scheme for extending the small-time approximating formulas to approximating formulas in a finite time horizon.

In this paper, we obtain a closed-form formula for an approximation to the optimal portfolio in a small-time horizon, where the asset price incorporates stochastic volatility with jumps.
We approximate the value function using the first-order terms of expansion of utility function in the power of time ($t$)  to the horizon ($T$), i.e., $ T - t $.  We control the error of this approximation using the second-order terms of expansion of utility function. After that, we generate a close-to-optimal portfolio near the time to horizon $ (T - t) $ using the first-order approximation of utility function. The error is also controlled by the square of the time to horizon $ (T - t)^2 $. Finally we provide an approximation scheme to the value function for all times $ t \in [0, T] $, and generate the close-to-optimal portfolio on $ [0, T] $.

The rest of the paper proceeds as follows. In Section \ref{sec2}, we provide the stochastic volatility model with jumps that are analyzed in this paper. We also describe model assumptions and utility assumptions in this section.  In Section \ref{sec3}, we prove the main result of this paper. In Section \ref{sec4}, we provide an approximation result of the main result that is described in Section \ref{sec3}. Numerical results are shown in this section. In Section \ref{sec5}, we analyze a close-to-optimal portfolio near the time to horizon $ T-t $ by the first-order approximation of utility function. The value function for all times $ t \in [0,T] $ is approximated in Section \ref{sec6}. A brief conclusion is provided in Section \ref{sec7}. Proofs of some technical lemmas are provided in Appendix \ref{apendix}.

\section{Market model and underlying assumptions}
\label{sec2}

We consider the following incomplete stochastic volatility market model, that is driven by both a Brownian motion and a jump process. We consider two assets: a risky asset and a risk-free asset. For convenience, we denote $ f_t = f(t) $, a univariate function of time $ t \in [0, T] $, where $ T $ is terminal time.
	We let the unit price of risk-free asset be $ 1 $. Then we define the unit price of risky asset $ S_t $ by
	\begin{align} \label{St}
			dS_t = S_t \Big[ \mu(Y_t)dt + \sigma(Y_t)dW^{(1)}_t + \int_\mathbb{R} \gamma_0(t,\zeta) \tilde{N}(dt,d\zeta) \Big], \quad S_0>0,
	\end{align}
	 where $ \mu(\cdot) $ is the mean rate of return and $ \sigma(\cdot) $ is the volatility. More descriptions on these functions will be provided in Assumption \ref{aiai}. We also define the stochastic factor $ Y_t $ by
	\begin{align} \label{Yt}
		dY_t = b(Y_t)dt + a(Y_t)\Big[\rho dW^{(1)}_t + \sqrt{1-\rho^2}\ dW^{(2)}_t\Big] +  \int_\mathbb{R} \gamma_2(t,\zeta) \tilde{N}(dt,d\zeta), \quad  Y_0 \in \mathbb{R},
	\end{align}
	where $ |\rho| < 1 $, and $ W^{(1)}_t $, $ W^{(2)}_t $ form a standard Brownian motion $ W_t = (W^{(1)}_t, W^{(2)}_t) $ which adapts to natural filtration $ \sigma(W_s: 0 \leq s \leq t)$.
	With the notation in \cite{Oksendal_1}, we denote the compensated jump measure $ \tilde{N}(\cdot) $ of L\'evy process $ \eta_t $ by $\tilde{N}(dt,d\zeta) \equiv N(dt,d\zeta)  -  \nu(d\zeta) \,dt$,
	where $ N(dt,d\zeta) $ is differential jump measure giving the number of jumps through $ dt $ with generic jump size $ d\zeta $, and $ \nu(\cdot) $ is L\'evy measure defined by $ \nu(d\zeta) = \mathbb{E}\big(N([0,1],d\zeta)\big) $ with expectation function $ \mathbb{E} $. The processes $ \gamma_0(t,\zeta) $ and $ \gamma_2(t,\zeta) $, as well as $ \gamma_1(t,\zeta) $ in follow-up definition, are predictable with respect to a filtration generated by the L\'evy process $ \eta_s $, $ s \leq t $.

	We let $ X_t $ be total wealth, and let $ \pi(t, X_t, Y_t) $ be the portfolio representing the discounted amount of $ X_t $ invested into risky asset. If $ \pi(t, X_t, Y_t) $ is self-financing, then we define the total wealth $ X_t $ by
	\begin{align} \label{Xt}
		dX_t = \sigma(Y_t) \pi(t, X_t, Y_t) \Big[ \lambda(Y_t)dt + dW^{(1)}_t + \int_\mathbb{R} \gamma_1(t,\zeta) \tilde{N}(dt,d\zeta) \Big], \quad X_0>0.
	\end{align}
Also, we denote the market price of risk by $ \lambda (Y_t) = \frac{\mu(Y_t) - R}{ \sigma(Y_t) }$, 
	where $ R $ is the constant interest rate of risk-free asset. 
	\begin{assumption} [Model Assumptions]
\label{aiai}
		Let $ C^k(\mathbb{R}) $, $ k \in \mathbb{N} $, be the space of the $ k $-th order continuously differentiable functions. It is assumed that the coefficients in stochastic differential equations \eqref{St} and \eqref{Yt}, as well as market price of risk $ \lambda $, satisfy:
		\begin{enumerate}
			\item $ \mu, \sigma, \gamma_0 \in C(\mathbb{R}) $; $ b \in C^1(\mathbb{R}) $; $ a, \gamma_2, \lambda \in C^2(\mathbb{R}) $;
			\item The $ \sigma $, $ \gamma_0 $, $ a $ and $ \gamma_2 $ are strictly positive;
			\item The $ b, b', a, 1/a, a', a'', \gamma_2, 1/\gamma_2, \partial_t \gamma_2, \partial_{tt} \gamma_2, \lambda, \lambda', \lambda'' $ are absolutely bounded.
		\end{enumerate}
	\end{assumption}
	
	Similar to Assumption 1 in \cite{Nadtochiy}, the above model assumptions ensure that the system of stochastic differential equations (SDEs) \eqref{St} and \eqref{Yt} has a unique strong solution.

	\begin{definition} [Admissible Portfolio] \label{def_1}
		A portfolio $ \pi_t = \pi(t, X_t, Y_t) $ is admissible if
		\begin{enumerate}
			\item it is progressively measurable with respect to natural filtration $ \sigma(W_s: 0 \leq s \leq t)$;
			\item it yields that $ X_t $ is strictly positive;
			\item for $ r >0 $, denoting $ \sigma_t = \sigma(Y_t) $,
			\begin{align*}
				& \mathbb{E} \bigg( \int_0^T \frac{\sigma^2_t \pi^2_t}{X^{2r}_t} dt \bigg) < \infty; \\
				& \mathbb{E} \bigg( \int_0^T \int_\mathbb{R} \Big[ \ln \big( X_t + \sigma_t\pi_t\gamma_1(t,\zeta) \big) - \ln(X_t) \Big]^2 \nu(d\zeta) dt \bigg) < \infty; \\
				& \mathbb{E} \bigg( \int_0^T \int_\mathbb{R} \Big[ \big[ X_t + \sigma_t\pi_t\gamma_1(t,\zeta) \big]^{1-r} - X^{1-r}_t \Big]^2 \nu(d\zeta) dt \bigg) < \infty, \ r \neq 1.
			\end{align*}
		\end{enumerate}
	\end{definition}

Let $ U = U(t,X_t,Y_t) $ be the \emph{utility function}, and $ J = J(t,X_t,Y_t) $ be the \emph{value function}. Over all the admissible portfolios, we define the value function $ J $ by
	\begin{align} \label{J}
		J(t,x,y) = \text{ess}\sup_{\pi(t,x,y)} \mathbb{E}\big( U_{_T}(X_{_T}) \big| X_t = x , Y_t = y \big),
	\end{align}
where the terminal condition $ U(T,x,y) = U(T,x) = U_{_T}(x) $. From the system of SDEs \eqref{Yt} and \eqref{Xt}, we obtain
	\begin{align*}
 \begin{bmatrix} dt \\ dx \\ dy \end{bmatrix}  &= \begin{bmatrix} 1 \\ \sigma(y) \pi(t,x,y) \lambda(y) \\ b(y) \end{bmatrix}dt 
		+ \begin{bmatrix} 0 \\ \sigma(y) \pi(t,x,y) \\ a(y) \rho \end{bmatrix}dW^{(1)}_t 
		+ \begin{bmatrix} 0 \\ 0 \\ a(y) \sqrt{1-\rho^2} \end{bmatrix}dW^{(2)}_t \\
		& \quad \
	 	+ \begin{bmatrix} 0 \\ \sigma(y) \pi(t,x,y) \\ 0 \end{bmatrix} 
	 	\int_\mathbb{R} \gamma_1(t,\zeta) \tilde{N}(dt,d\zeta)
	 	+ \begin{bmatrix} 0 \\ 0 \\ 1 \end{bmatrix} 
	 	\int_\mathbb{R} \gamma_2(t,\zeta) \tilde{N}(dt,d\zeta).
	\end{align*}

Following we summarize the Theorem 3.2 from \cite{Oksendal_1}. This theorem will be utilized for the rest of the paper. 
\begin{theorem} \label{Oksendal}
	\begin{enumerate}
		\item Suppose we can find a function $ \varphi \in \mathcal{C}^2 (\mathbb{R}^n) $ such that
		\begin{enumerate}
			\item $ A_v \varphi(y) + f (y, v) \leq 0 $, for all $ v \in \mathcal{V} $, where $ \mathcal{V} $ is the set of possible control values, and
			\begin{align} \label{generator}
				A_v \varphi(y) =\ 
				& \sum_{i=1}^k b_i(y,v) \frac{\partial \varphi}{\partial y_i}(y) + \frac{1}{2} \sum_{i,j=1}^k (\sigma \sigma^T)_{ij}(y,v) \frac{\partial^2 \varphi}{\partial y_i \partial y_j} \nonumber \\
				& + \sum_m \int_\mathbb{R} \{ \varphi(y+\gamma^{(k)}(y,v,\zeta)) - \varphi(y) - \nabla  \varphi(y) \gamma^{(k)}(y,v,\zeta) \} \nu_k (d\zeta)
			\end{align}
			\item $ \lim_{t\rightarrow \tau_\mathcal{S}} \varphi(Y(t)) = g(Y(\tau_\mathcal{S})) 1_{\{\tau_\mathcal{S}<\infty\}} $
			\item ``Growth conditions."
			\begin{align}
				\mathbb{E}^y\Big[ 
				& |\varphi(Y(\tau))| + \int_0^{\tau_\mathcal{S}} \{ |A \varphi(Y(t))| + |\sigma^T(Y(t))\nabla \varphi(Y(t))|^2 \nonumber \\
				& + \sum_{j=1}^l \int_\mathbb{R} |\varphi(Y(t)+\gamma^{(j)}(Y(t),u(t),\zeta_j)) - \varphi(Y(t))|^2 \nu_j (d\zeta_j)
				\}dt
				\Big] <\infty, \label{growth}
			\end{align}
			for all $ u \in \mathcal{A} $ and all stopping time $ \tau $.
			\item $ \{ \varphi^- (Y(\tau)) \}_{\tau \leq \tau_\mathcal{S}} $ is uniformly integrable for all $ u \in \mathcal{A} $ and $ y \in \mathcal{S} $, where, in general, $ x^- := max\{-x, 0\} $ for $ x \in \mathbb{R} $. Then $\varphi(y) \geq \Phi(y)$. 
		\end{enumerate}
		\item Suppose we for all $ y \in \mathcal{S} $ can find $ v = \widehat{u}(y) $ such that $A_{\widehat{u}(y)} \varphi(y) + f(y, \widehat{u}(y)) = 0$, 
		and $ \widehat{u}(y) $ is an admissible feedback control (Markov control), i.e. $ \widehat{u}(y) $
		means $ \widehat{u}(Y(t)) $. Then $ v = \widehat{u}(y) $ is an optimal control and $\varphi(y) = \Phi(y)$. 
	\end{enumerate}
\end{theorem}

	By \eqref{generator}, we obtain the generator of $ U $ is given by
	\begin{align*}
	 	& \partial_t U + \lambda(y)\sigma(y)\pi(t,x,y)U_x + b(y)U_y \\
	 	& + \frac{1}{2} \sigma^2(y)\pi^2(t,x,y) U_{xx} + \rho a(y)\sigma(y)\pi(t,x,y) U_{xy} + \frac{1}{2}a^2(y)U_{yy} \\
	 	& + \int_\mathbb{R} \Big[ U\big(t, x+ \sigma(y)\pi(t,x,y)\gamma_1(t,\zeta), y+\gamma_2(t,\zeta)\big) - U  \\
	 	& \quad \quad \quad 
	 	- \sigma(y)\pi(t,x,y)\gamma_1(t,\zeta)U_x - \gamma_2(t,\zeta)U_y \Big] \nu(d\zeta).
	\end{align*}
	Consequently, by  Theorem \ref{Oksendal},  we find that $ J(t,x,y) $  is the solution of HJB integro-differential equation
	\begin{align} \label{HJB_0}
	 	& \max_{\pi(t,x,y)} \bigg(
	 	\lambda(y)\sigma(y)\pi(t,x,y)U_x + \frac{1}{2}\sigma^2(y)\pi^2(t,x,y)U_{xx} + \rho a(y)\sigma(y)\pi(t,x,y) U_{xy} \nonumber\\
	 	& \quad \quad \quad \
	 	+ \int_\mathbb{R} \Big[ U\big(t, x+ \sigma(y)\pi(t,x,y)\gamma_1(t,\zeta), y+\gamma_2(t,\zeta)\big)
	 	 - \sigma(y)\pi(t,x,y)\gamma_1(t,\zeta) U_x \Big] \nu(d\zeta) 
	 	\bigg) \nonumber\\
	 	& + \partial_t U  + b(y)U_y + \frac{1}{2}a^2(y)U_{yy} 
	 	- \int_\mathbb{R} \Big[U + \gamma_2(t,\zeta) U_y \Big] \nu(d\zeta)= 0,
	\end{align}
	with terminal condition $ U(T,x,y) = U(T,x) = U_{_T}(x) $. We assume that $ U_{_T}(x) $ has the properties as below.
	\begin{assumption} [Utility Assumptions] \label{Assume_UT}
	 	It is assumed that $ U_{_T}(x) $ is strictly increasing and concave in $ C^5(\mathbb{R}) $, and $ U_{_T}(x) $ behaves asymptotically when $ x $ approaches to both $ 0 $ and $ \infty $, i.e.\ $ U_{_T}(x) $ satisfies either of the following cases:
	 	\begin{enumerate}
	 		\item Logarithmic function: $ U_{_T}(x) = \ln (x) $.
	 		\item Mixture of power functions: $ U_{_T}(x) = \frac{c_1}{1-\alpha}x^{1-\alpha} + \frac{c_2}{1-\beta}x^{1-\beta} $, where $ c_1,c_2 \geq 0 $, $ \alpha, \beta > 0 $, and $ \alpha, \beta \neq 1 $.
	 	\end{enumerate}
	 \end{assumption}
	 
	 Because $ U\big(t, x+ \sigma(y)\pi(t,x,y)\gamma_1(t,\zeta), y+\gamma_2(t,\zeta)\big) $ is an implicit function of $ \pi(t,x,y) $,  we consider to expand $ U $ at terminal time $ T $ using power series: 
	\begin{align} \label{U}
		U(t,x,y) :=  U_{_T}(x) + (T-t)U^{(1)}(x,y) + (T-t)^2U^{(2)}(x,y),
	\end{align}
	where $ U^{(0)}(x,y) = U(T,x,y) = U_{_T}(x) $.

Consequently, we have
	\begin{align} \label{Uxy}
		& U\big(t, x+ \sigma(y)\pi(t,x,y)\gamma_1(t,\zeta), y+\gamma_2(t,\zeta)\big) \nonumber \\
		& = U_{_T}(\chi) + (T-t)U^{(1)}(\chi, \psi) + (T-t)^2U^{(2)}(\chi, \psi),
	\end{align}
	where
	\begin{align}
		\chi & = x+ \sigma(y)\pi_{_T}(x,y)\gamma_{1_T}(\zeta), \label{chi} \\ 
		\psi & = y+\gamma_{2_T}(\zeta). \nonumber
	\end{align}
	Then we substitute \eqref{Uxy} into \eqref{HJB_0}. If the system of SDEs \eqref{Yt} and \eqref{Xt} is Markovian, then the expression being maximized in HJB equation \eqref{HJB_0} achieves its maximum at the optimal portfolio denoted by $ \hat{\pi}(t,x,y) $. By the first-order condition, we have
	\begin{align} \label{pi}
		\hat{\pi}(t,x,y) = 
		\frac{ -\lambda(y)U_x - \rho a(y) U_{xy}} 
		{\sigma(y)U_{xx}}
		+ \frac{U_x}{\sigma(y)U_{xx}}\int_\mathbb{R}  \gamma_1(t, \zeta) \nu(d\zeta).
	\end{align}
	After substituting \eqref{Uxy} and \eqref{pi} into \eqref{HJB_0}, we obtain the following HJB equation
	\begin{align} \label{HJB}
		\left\{
			\begin{array}{ll}
				\partial_t U + \mathcal{H}(U)  =0 & \text{for } (t,x,y) \in [0,T]\times \mathbb{R}^+ \times \mathbb{R} \\
				U(T,x,y) = U_{_T}(x) & \text{for } (x,y) \in \mathbb{R}^+ \times \mathbb{R} , \\
			\end{array}
		\right.
	\end{align}
	where
	\begin{align*}
		\mathcal{H}(U) =\
		& \lambda(y)\sigma(y)\hat{\pi}(t, x, y)U_x + b(y)U_y \\ 
		& + \frac{1}{2}\sigma^2(y)\hat{\pi}^2(t, x, y)U_{xx} + \rho a(y)\sigma(y)\hat{\pi}(t, x, y) U_{xy} +\frac{1}{2}a^2(y)U_{yy} \\
		& + \int_\mathbb{R} \Big[ U_{_T}(\hat{\chi}) + (T-t)U^{(1)}(\hat{\chi}, \psi) + (T-t)^2U^{(2)}(\hat{\chi}, \psi) - U \\
		& \quad \quad \quad 
		- \sigma(y)\hat{\pi}(t, x, y)\gamma_1(t,\zeta) U_x - \gamma_2(t,\zeta) U_y \Big] \nu(d\zeta),
	\end{align*}
	 with $ \hat{\chi} = x+ \sigma(y)\hat{\pi}_{_T}(x,y)\gamma_{1_T}(\zeta) $.

	\section{Main results}
\label{sec3}
	In this section, we construct classical super-solution and sub-solution to HJB equation \eqref{HJB} using the second order expansion of utility function in powers of the time to horizon $ T - t $.
	We approximate the value function defined by \eqref{J} using the first order terms of expansion in power of the time to horizon $ T - t $. 
	We then control the error of this approximation by the second order terms of expansion in powers of the square time to horizon $ (T - t)^2 $.
	We also prove that value function lies between constructed super-solution and sub-solution using martingale inequalities. 
	
	\begin{theorem} \label{thm}
		Let $ \hat{U}(t,x,y) $ be the solution of HJB equation \eqref{HJB}. Then
		\begin{align} \label{U_hat}
			\hat{U}(t,x,y) = U_{_T}(x) + (T-t)U^{(1)}(x,y),
		\end{align}
		where
		\begin{align} \label{U^1}
			U^{(1)}(x,y) = \
			& \frac{1}{2}\bigg[ -\lambda^2(y)+\Big[\int_\mathbb{R} \gamma_{1_T}(\zeta) \nu(d\zeta)\Big]^2 \bigg] \frac{\big[U_{_T}'(x) \big]^2}{ U_{_T}''(x)} \nonumber\\
			& + \int_\mathbb{R} \bigg[  U_{_T}(\hat{\chi}) - U_{_T}(x)
			 + \Big[\lambda(y) - \int_\mathbb{R} \gamma_{1_T}(\zeta) \nu(d\zeta)\Big] \frac{\big[U_{_T}'(x)\big]^2}{U_{_T}''(x)} \gamma_{1_T}(\zeta) \bigg]\nu(d\zeta),
		\end{align}
		and
		\begin{align} 
			 \hat{\chi} = x - \Big[\lambda(y) - \int_\mathbb{R} \gamma_{1_T}(\zeta) \nu(d\zeta)\Big] \frac{U_{_T}'(x)}{U_{_T}''(x)} \gamma_{1_T}(\zeta) \label{chi_1}.
		\end{align}
		Moreover, there exists constants $ c > 0 $ and $ 0 < \varepsilon < \min\{1, T\} $ such that
		\begin{align} \label{ineq}
			\big| J(t,x,y) - \hat{U}(t,x,y) \big| \leq c (T-t)^2 f(x) \quad \text{for }  (t,x,y) \in (T-\varepsilon,T)\times \mathbb{R}^+ \times \mathbb{R},
		\end{align}
		where $ J(t,x,y) $ is value function defined by \eqref{J}; $ f(x) = 1 $ under Case 1 of Assumption \ref{Assume_UT} and $ f(x) = x^{1-\alpha} + x^{1-\beta} $ under Case 2 of Assumption \ref{Assume_UT}; the constants $ c $ and $ \varepsilon$ are independent of $ t $, $ x $ and $ y $.
	\end{theorem}
	\begin{proof} We prove this theorem in two steps. \\
		\textbf{Step 1.} We construct super- and sub-solutions to HJB equation \eqref{HJB}. Similar to \eqref{U}, we expand $ \gamma_1(t,\zeta) $ and $ \gamma_2(t,\zeta) $ at terminal time $ T $ using power series:
		\begin{align} 
			\gamma_1(t,\zeta) & = \gamma_{1_T}(\zeta) + (T-t)\gamma_1^{(1)}(\zeta) + (T-t)^2\gamma_1^{(2)}(\zeta) \label{gam1}, \\ 
			\gamma_2(t,\zeta) & = \gamma_{2_T}(\zeta) + (T-t)\gamma_2^{(1)}(\zeta) + (T-t)^2\gamma_2^{(2)}(\zeta) \label{gam2}.
		\end{align}
		For convenience, we denote $ \lambda = \lambda(y), \sigma = \sigma(y), b = b(y), a = a(y) $, $ \gamma_{1_T} = \gamma_{1_T}(\zeta) $, $ \gamma_{2_T} = \gamma_{2_T}(\zeta) $, $ \gamma_1^{(1)} = \gamma_1^{(1)}(\zeta) $, $ \delta = T-t $, $ U^{(1)} = U^{(1)}(x,y) $, $ U^{(2)} = U^{(2)}(x,y) $ and $ \int \cdot = \int_\mathbb{R} \cdot\ \nu(d\zeta) $.
		
		We substitute \eqref{U}, \eqref{gam1} and \eqref{gam2} into HJB equation \eqref{HJB} with optimal portfolio \eqref{pi}. Then by ignoring the terms of order $ O(2) $, we have
		\begin{align} \label{HJB_2}
			& -U^{(1)} - 2\delta U^{(2)}
			+ \lambda\sigma\hat{\pi}(t, x, y)\Big[U_{_T}'(x) + \delta U_x^{(1)}\Big]
			+ b \delta U_y^{(1)} \nonumber\\
			& + \frac{1}{2}\sigma^2\hat{\pi}^2(t, x, y)\Big[U_{_T}''(x) + \delta U_{xx}^{(1)}\Big]
			+ \rho a \sigma \hat{\pi}(t, x, y) \delta U_{xy}^{(1)}
			+\frac{1}{2}a^2 \delta U_{yy}^{(1)} \nonumber\\
			& + \int \bigg[ U_{_T}(\hat{\chi}) + \delta U^{(1)}(\hat{\chi}, \psi) - U_{_T}(x) - \delta U^{(1)} \nonumber\\
			& \quad \quad \quad 
			- \sigma \hat{\pi}(t, x, y) \Big[U_{_T}'(x)\gamma_{1_T} + \delta U_x^{(1)}\gamma_{1_T} + \delta U_{_T}'(x)\gamma_1^{(1)} \Big] - \delta U_y^{(1)}\gamma_{2_T} \bigg] =0,
		\end{align}
		where
		\begin{align} \label{pi_2}
			\hat{\pi}(t,x,y) = \ 
			& \frac{- \lambda\big[U_{_T}'(x) + \delta U_x^{(1)}\big] - \rho a \delta U_{xy}^{(1)}}{\sigma\big[U_{_T}''(x) + \delta U_{xx}^{(1)}\big]} \nonumber \\
			& + \frac{1}
			{\sigma\big[U_{_T}''(x) + \delta U_{xx}^{(1)}\big]}
			\int \Big[U_{_T}'(x)\gamma_{1_T} + \delta U_x^{(1)}\gamma_{1_T} + \delta U_{_T}'(x)\gamma_1^{(1)} \Big].
		\end{align}
		By ignoring the terms of order $ O(1) $ in equation \eqref{HJB_2}, we have
		\begin{align} \label{U1}
			U^{(1)} = \ 
			& \lambda\sigma\hat{\pi}_{_T}(x, y) U_{_T}'(x)
			+ \frac{1}{2}\sigma^2\hat{\pi}_{_T}^2(x, y) U_{_T}''(x) + \int \Big[ U_{_T}(\hat{\chi}) - U_{_T}(x) - \sigma\hat{\pi}_{_T}(x, y) U_{_T}'(x)\gamma_{1_T} \Big],
		\end{align}
		where
		\begin{align} \label{pi_3}
			\hat{\pi}_{_T}(x,y) = \frac{- \lambda + \int \gamma_{1_T}} {\sigma } \cdot \frac{U_{_T}'(x)}{U_{_T}''(x)}.
		\end{align}
		Finally we substitute \eqref{pi_3} into \eqref{U1} and substitute \eqref{pi_3} into \eqref{chi}, proving the equations \eqref{U^1} and \eqref{chi_1} respectively. 
		
		We substitute \eqref{pi_2} into \eqref{HJB_2} and clear fractions. Then we equate the sum of coefficients of $ \delta $ to zero and solve the $ U^{(2)} $ as
		\begin{align}\label{U^2}
			U^{(2)} =\ 
			& \frac{1}{2}  \bigg[ b U_y^{(1)} + \frac{1}{2} a^2  U_{yy}^{(1)} + \int \Big[  U^{(1)}(\hat{\chi},\psi) -  U^{(1)}  -  U_y^{(1)}\gamma_{2_T} \Big] \bigg] \nonumber\\
			& +    \frac{U_{xx}^{(1)}}{U_{_T}''(x)} \bigg[ -U^{(1)} + \int \Big[ U_{_T}(\hat{\chi}) - U_{_T}(x) \Big] \bigg] \nonumber\\
			& + \frac{L}{2}  \frac{U_{_T}'(x)}{U_{_T}''(x)}  \bigg[ - \lambda  U_x^{(1)} - \rho a  U_{xy}^{(1)} +  \int \Big[ U_x^{(1)}\gamma_{1_T} +  U_{_T}'(x)\gamma_1^{(1)} \Big] \bigg] - \frac{L^2}{4} \frac{U_{xx}^{(1)}}{U_{_T}''(x)} \frac{\big[U_{_T}'(x)\big]^2}{U_{_T}''(x)},
		\end{align}
		where $ L = \lambda - \int  \gamma_{1_T} $, 
		\begin{align}
			U^{(1)}_x =\
			& \frac{1}{2}\Big[ -\lambda^2+\big[\int \gamma_{1_T} \big]^2 \Big] \Big(\frac{\big[U_{_T}'(x) \big]^2}{U_{_T}''(x)}\Big)' \nonumber\\
			& + \int \bigg[ U_{_T}'(\hat{\chi}) \Big[1 -  L \Big(\frac{U_{_T}'(x)}{U_{_T}''(x)}\Big)' \gamma_{1_T}\Big]
			- U_{_T}'(x) + L \Big(\frac{\big[U_{_T}'(x) \big]^2}{U_{_T}''(x)}\Big)' \gamma_{1_T} \bigg], \label{U1x} \\
			U^{(1)}_{xx} =\
			& \frac{1}{2}\Big[ -\lambda^2+\big[\int \gamma_{1_T} \big]^2 \Big] 
			\Big(\frac{\big[U_{_T}'(x) \big]^2}{U_{_T}''(x)}\Big)'' \nonumber\\
			& + \int \bigg[ U_{_T}''(\hat{\chi}) \Big[ 1 -  L \Big(\frac{U_{_T}'(x)}{U_{_T}''(x)}\Big)' \gamma_{1_T} \Big]^2  -  L U_{_T}'(\hat{\chi})  \Big(\frac{U_{_T}'(x)}{U_{_T}''(x)}\Big)'' \gamma_{1_T} - U_{_T}''(x)  \nonumber\\
			& \quad \quad \quad
			+ L \Big(\frac{\big[U_{_T}'(x) \big]^2}{U_{_T}''(x)}\Big)'' \gamma_{1_T} \bigg], \label{U1xx} \\
			U^{(1)}_y =\
			& -\lambda\lambda'  \frac{\big[U_{_T}'(x) \big]^2}{U_{_T}''(x)} 
			+ \int \lambda' \frac{\big[U_{_T}'(x)\big]^2}{U_{_T}''(x)} \gamma_{1_T}, \label{Uy}\\
			U^{(1)}_{yy} =\
			& -\Big[(\lambda')^2+ \lambda\lambda''\Big]  \frac{\big[U_{_T}'(x) \big]^2}{U_{_T}''(x)}
			+ \int \lambda'' \frac{\big[U_{_T}'(x)\big]^2}{U_{_T}''(x)} \gamma_{1_T}, \label{Uyy} \\
			U^{(1)}_{xy} =\
			& -\lambda\lambda' \Big(\frac{\big[U_{_T}'(x) \big]^2}{U_{_T}''(x)}\Big)'
			+ \int  \lambda' \Big(\frac{\big[U_{_T}'(x) \big]^2}{U_{_T}''(x)}\Big)' \gamma_{1_T}, \label{U_xy}
		\end{align}
		with
		\begin{align}
			& \Big(\frac{U_{_T}'(x)}{U_{_T}''(x)}\Big)'
			= 1 - \frac{U_{_T}'(x)}{\big[U_{_T}''(x) \big]^2} U_{_T}^{(3)}(x), \label{key1'} \\
			& \Big(\frac{U_{_T}'(x)}{U_{_T}''(x)}\Big)''
			= - \frac{U_{_T}^{(3)}(x)}{U_{_T}''(x)} + 2 \frac{U_{_T}'(x)}{\big[U_{_T}''(x) \big]^3} \big[ U_{_T}^{(3)}(x) \big]^2 - \frac{U_{_T}'(x)}{\big[U_{_T}''(x) \big]^2} U_{_T}^{(4)}(x), \label{key1''} \\
			& \Big(\frac{\big[U_{_T}'(x) \big]^2}{U_{_T}''(x)}\Big)'
			= 2U_{_T}'(x) -  \Big[\frac{U_{_T}'(x)}{U_{_T}''(x)} \Big]^2 U_{_T}^{(3)}(x), \label{key2'} \\
			& \Big(\frac{\big[U_{_T}'(x) \big]^2}{U_{_T}''(x)}\Big)'' 
			= 2U_{_T}''(x) - 2\frac{U_{_T}'(x)}{U_{_T}''(x)}U_{_T}^{(3)}(x) + 2\frac{\big[U_{_T}'(x)\big]^2}{\big[U_{_T}''(x)\big]^3}\big[U_{_T}^{(3)}(x)\big]^2
			- \Big[\frac{U_{_T}'(x)}{U_{_T}''(x)} \Big]^2 U_{_T}^{(4)}(x). \label{key2''}
		\end{align}
		We take $ U_{_T}(\hat{\chi}) - U_{_T}(x) $ as one term, and enumerate all the terms of $ U^{(2)} $ as $ u_1^{(2)} , \cdots , u_l^{(2)} $, that is $ U^{(2)} = \sum_{i=1}^l u_i^{(2)} $. We then prove Lemma \ref{A.0} in the Appendix \ref{apendix} showing $ u_i^{(2)} \sim f(x) $ for $ 1\leq i \leq l $, where $ f(x) = 1 $ under Case 1 of Assumption \ref{Assume_UT} and $ f(x) = x^{1-\alpha} + x^{1-\beta} $ under Case 2 of Assumption \ref{Assume_UT}.
		
		We let
		\begin{align*}
		u^{(2)}(x,y) = 1 + l \cdot \max_{1 \leq i \leq l} \sup_{x>0} \frac{\big|u_i^{(2)}(x,y)\big|}{f(x)}.
		\end{align*}
		Then we define super-solution $ \overline{U} = \overline{U}(t,x,y) $ and sub-solution $ \underline{U} = \underline{U}(t,x,y) $ to HJB equation \eqref{HJB} by 
		\begin{align*}
			\overline{U}(t,x,y) & = U_{_T}(x) + (T-t)U^{(1)}(x,y) + (T-t)^2u^{(2)}(x,y)f(x), \\
			\underline{U}(t,x,y) & = U_{_T}(x) + (T-t)U^{(1)}(x,y) - (T-t)^2u^{(2)}(x,y)f(x).
		\end{align*}
		We note that the coefficient of $ U^{(2)} $ is $ - 2(T-t) < 0 $ in \eqref{HJB_2}. After substituting $ \overline{U} $ into the left hand side of HJB equation \eqref{HJB} and clearing fractions, we observe that the coefficient of $ T-t $ is strictly negative based on definition of $ \overline{U} $. On the other hand, after substituting $ \underline{U} $, we observe that the coefficient of $ T-t $ is strictly positive based on  definition of $ \underline{U} $.
		
		Because $ U^{(2)} $ is solved by equating the sum of coefficients of $ T-t $ to zero, therefore $ u_i^{(2)} \sim f(x) $ implies the coefficient of $ T-t $ in either $ \overline{U}_t + \mathcal{H}(\overline{U}) $ or $ \underline{U}_t + \mathcal{H}(\underline{U}) $ is in the order of $ f(x) $.
		We recall equation \eqref{sim7} showing $ U^{(1)}_{xx} \sim U_{_T}''(x) $, which implies $ U_{_T}''(x)U_{xx}^{(1)} \sim \big[ U_{_T}''(x) \big]^2 $. Then the coefficient of $ T-t $ in either $ \big[\overline{U}_{xx}\big]^2 $ or $ \big[\underline{U}_{xx}\big]^2 $ is in the order of $ \big[ U_{_T}''(x) \big]^2 $. 
		Thus, as inequality (3.8) in \cite{Kumar}, we still have the following inequalities: 
		\begin{align} 
			\big[\overline{U}_{xx}\big]^2 \big|\overline{U}_t + \mathcal{H}(\overline{U})\big| & \leq c_3 (T-t)\tilde{f}(x), \label{3.8}\\
			\big[\underline{U}_{xx}\big]^2 \big|\underline{U}_t + \mathcal{H}(\underline{U})\big| & \leq c_4 (T-t)\tilde{f}(x), \nonumber
		\end{align}
		where $ c_3 $ and $ c_4 $ are constants, and $ \tilde{f}(x) \sim \big[ U_{_T}''(x) \big]^2f(x) $.
		Under Case 1 of Assumption \ref{Assume_UT}, we have $ \tilde{f}(x) = x^{-4} $.
		Under Case 2 of Assumption \ref{Assume_UT}, we have
		\begin{align*}
			\lim_{x \rightarrow \infty} \frac{\big[U_{_T}''(x) \big]^2}{x^{-2\alpha-2} + x^{-2\beta-2}} = 
			\lim_{x \rightarrow \infty} \frac{[-c_1\alpha x^{-\alpha-1} - c_2\beta x^{-\beta-1}]^2}{x^{-2\alpha-2} + x^{-2\beta-2}}
			= c_1^2 \alpha^2 \text{ or } c_2^2\beta^2,
		\end{align*}
		that is $ \big[U_{_T}''(x) \big]^2 \sim x^{-2\alpha-2} + x^{-2\beta-2} $ which gives $ \tilde{f}(x) = [x^{-2\alpha-2} + x^{-2\beta-2}][x^{1-\alpha} + x^{1-\beta}] $. Furthermore, for some $ c_3 <0 $, we observe that $ \frac{1}{c_3}x^{-2} \leq \overline{U}_{xx} \leq c_3x^{-2} < 0 $ under Case 1 of Assumption \ref{Assume_UT} and $ \frac{1}{c_3} [x^{-\alpha-1} + x^{-\beta-1}] \leq \overline{U}_{xx} \leq c_3[x^{-\alpha-1} + x^{-\beta-1}] <0 $ under Case 2 of Assumption \ref{Assume_UT}. Thus $ \big[\overline{U}_{xx}\big]^2 $ is bounded away from zero. So is the $ \big[\underline{U}_{xx}\big]^2 $ for some $ c_4 <0 $.
		
		We already have the coefficient of $ T-t $ in $ \big[\overline{U}_{xx}\big]^2 \big[\overline{U}_t + \mathcal{H}(\overline{U})\big] $ is strictly negative, in the order of $ \tilde{f}(x) $, and bounded. In addition, we observe that inequality \eqref{3.8} implies the $ o(T-t )$ terms in $ \big[\overline{U}_{xx}\big]^2 \big[\overline{U}_t + \mathcal{H}(\overline{U})\big] $ are in the order of $ \tilde{f}(x) $ and bounded. For $ t $ near $ T $, the strictly negative coefficient of $ T-t $ uniformly dominates the $ o(T-t ) $ terms, that is $ \overline{U}_t + \mathcal{H}(\overline{U}) < 0 $. Thus $ \overline{U} $ is the classical super-solution of HJB equation \eqref{HJB}. To prove $ \underline{U} $ is the classical sub-solution of HJB equation \eqref{HJB}, we can use a mirror of this discussion. \\
		\textbf{Step 2.} We prove $ \underline{U}(t,x,y) \leq J(t,x,y) \leq \overline{U}(t,x,y) $.
		We substitute super-solution $ \overline{U} $ into optimal portfolio \eqref{pi} to generate the portfolio
		\begin{align*}
			\overline{\hat{\pi}} = \overline{\hat{\pi}}(t,x,y)
			=
			\frac{-\lambda(y)\overline{U}_x - \rho a(y) \overline{U}_{xy}}{\sigma(y)\overline{U}_{xx}}
			+ \frac{\overline{U}_x}{\sigma(y)\overline{U}_{xx}} \int_\mathbb{R}  \gamma_1(t, \zeta) \nu(d\zeta).
		\end{align*}
		Then we apply the 2-dimensional It\^{o}'s formula to $ \overline{U} $ and obtain
		\begin{align}
			& \overline{U}(T, X_{_T}, Y_{_T}) - \overline{U}(t,x,y)  \nonumber\\
			& = \int_t^T \Big[ \partial_s \overline{U} + \lambda\sigma\overline{\hat{\pi}}\overline{U}_x + b\overline{U}_y
			+ \frac{1}{2}\sigma^2\overline{\hat{\pi}}^2\overline{U}_{xx} + \rho a\sigma\overline{\hat{\pi}} \overline{U}_{xy} +\frac{1}{2}a^2\overline{U}_{yy} \Big] ds \nonumber\\
			& \quad
			+ \int_t^T \Big[ \sigma\overline{\hat{\pi}}\overline{U}_x + \rho a \overline{U}_y \Big] dW^{(1)}_s
			+ \int_t^T \Big[ a\sqrt{1-\rho^2}\overline{U}_y \Big] dW^{(2)}_s \label{mart_1}\\
			& \quad
			+ \int_t^T \int_\mathbb{R} \Big[ \overline{U}\big(s, x+ \sigma\overline{\hat{\pi}}\gamma_1(s,\zeta), y+\gamma_2(s,\zeta)\big) - \overline{U} - \sigma\overline{\hat{\pi}}\gamma_1(s,\zeta) \overline{U}_x 
			- \gamma_2(s,\zeta) \overline{U}_y \Big] \nu(d\zeta) ds \nonumber\\
			& \quad
			+ \int_t^T \int_\mathbb{R} \Big[ \overline{U}\big(s, x+ \sigma\overline{\hat{\pi}}\gamma_1(s,\zeta), y+\gamma_2(s,\zeta)\big) - \overline{U} \Big]\tilde{N}(ds,d\zeta). \label{mart_2}
		\end{align}
		We note that the stochastic integrals \eqref{mart_1} and \eqref{mart_2} are local martingales. Then we let $ \{t_n \}_{n=1}^\infty \subset [t, T] $ be a sequence of stopping times such that $ t_n \leq t_{n+1} $ and $ t_n \rightarrow T $. With replacing $ T $ by $ t_n $, we observe that the stochastic integrals \eqref{mart_3} and \eqref{mart_4} are martingales. 
		\begin{align}
			& \overline{U}(t_n, X_{t_n}, Y_{t_n}) - \overline{U}(t,x,y)  \nonumber\\
			& = \int_t^{t_n} \Big[ \partial_s \overline{U} + \lambda\sigma\overline{\hat{\pi}}\overline{U}_x + b\overline{U}_y
			+ \frac{1}{2}\sigma^2\overline{\hat{\pi}}^2\overline{U}_{xx} + \rho a\sigma\overline{\hat{\pi}} \overline{U}_{xy} +\frac{1}{2}a^2\overline{U}_{yy} \Big] ds \label{HJB_3}\\
			& \quad
			+ \int_t^{t_n} \Big[ \sigma\overline{\hat{\pi}}\overline{U}_x + \rho a \overline{U}_y \Big] dW^{(1)}_s
			+ \int_t^{t_n} \Big[ a\sqrt{1-\rho^2}\overline{U}_y \Big] dW^{(2)}_s \label{mart_3}\\
			& \quad
			+ \int_t^{t_n} \int_\mathbb{R} \Big[ \overline{U}\big(s, x+ \sigma\overline{\hat{\pi}}\gamma_1(s,\zeta), y+\gamma_2(s,\zeta)\big) - \overline{U} - \sigma\overline{\hat{\pi}}\gamma_1(s,\zeta)\overline{U}_x  
			- \gamma_2(s,\zeta)\overline{U}_y \Big] \nu(d\zeta) ds \label{HJB_4}\\
			& \quad
			+ \int_t^{t_n} \int_\mathbb{R} \Big[ \overline{U}\big(s, x+ \sigma\overline{\hat{\pi}}\gamma_1(s,\zeta), y+\gamma_2(s,\zeta)\big) - \overline{U} \Big]\tilde{N}(ds,d\zeta). \label{mart_4}
		\end{align}
		Furthermore, we note that the integrand of \eqref{HJB_3} $ + $ \eqref{HJB_4} is exactly the $ \overline{U}_t + \mathcal{H}(\overline{U}) < 0 $. By the martingale property that conditional expectation of martingale is zero, we have $ \mathbb{E}\big(\overline{U}(t_n, X_{t_n}, Y_{t_n}) - \overline{U} \big| x, y \big) < 0 $, that is
		\begin{align} \label{ineq2}
			 \mathbb{E}\big(\overline{U}(t_n, X_{t_n}, Y_{t_n}) \big| x, y \big) < \overline{U}.
		\end{align}
		
		By definition of $ \overline{U} $ and triangle inequality, we have
		\begin{align*}
			\big| \overline{U}(t_n, X_{t_n}, Y_{t_n}) \big| 
			& = \big| U_{_T}(X_{t_n}) + (T-t_n)U^{(1)}(X_{t_n},Y_{t_n}) + (T-t_n)^2u^{(2)}(X_{t_n},Y_{t_n}) f(X_{t_n}) \big| \\
			& \leq \big| U_{_T}(X_{t_n}) \big| + \big| T U^{(1)}(X_{t_n},Y_{t_n}) \big| + \big| T^2u^{(2)}(X_{t_n},Y_{t_n}) f(X_{t_n}) \big| \\
			& \leq c_5 g(X_{t_n}),
		\end{align*}
		with some constant $ c_5 $. We recall equation \eqref{sim3} showing $ U^{(1)} \sim f(x) $, and observe that $ u_j^{(2)} \sim f(x) \Rightarrow u^{(2)}(x,y) f(x) \sim f(x) $. Then we have $ g(x) = U_{_T}(x) + f(x) = \ln(x) +1 $ under Case 1 of Assumption \ref{Assume_UT}.
		In addition, we recall equation \eqref{sim12} showing $ U_{_T}(x) \sim f(x) $ under Case 2 of Assumption \ref{Assume_UT}. Thus we have $ g(x) = f(x) = x^{1-\alpha} + x^{1-\beta} $ under this case.
		
		We prove Lemma \ref{A.1} in the Appendix \ref{apendix} showing that $ \big\{ g(X_{t_n}) \big\}_{n=1}^\infty $ is uniformly bounded by an integrable random variable. We already have $ \big| \overline{U}(t_n, X_{t_n}, Y_{t_n}) \big| \leq c_5 g(X_{t_n}) $. And we observe that 
		\begin{align*}
			t_n \rightarrow T\ \Rightarrow\ \overline{U}(t_n, X_{t_n}, Y_{t_n}) \rightarrow \overline{U}(T, X_{_T}, Y_{_T}) = U_{_T}(X_{_T}).
		\end{align*}
		By Dominated Convergence Theorem, we have
		\begin{align} \label{lim}
			\lim_{n \rightarrow \infty} \mathbb{E}\big(\overline{U}(t_n, X_{t_n}, Y_{t_n}) \big| x, y \big) = \mathbb{E}\big(U_{_T}(X_{_T}) \big| x, y \big).
		\end{align}
		Combining \eqref{ineq2} with \eqref{lim}, we have $ \mathbb{E}\big(U_{_T}(X_{_T}) \big| x, y \big) < \overline{U} $ . 
		We also prove Lemma \ref{A.2} in the Appendix \ref{apendix} showing that $ \overline{\hat{\pi}} $ is admissible. Thus $ J \leq \overline{U} $. 
		
		To prove $ J \geq \underline{U} $, we can use a mirror of above discussions. Finally by definitions of $ \underline{U} $ and $ \overline{U} $, the inequality \eqref{ineq} holds.
	\end{proof}

\section{Numerical examples}
\label{sec4}

In this section, we introduce a further approximation to achieve the solution given by \eqref{U_hat} for HJB equation \eqref{HJB} and find the rate convergence of such approximation. We choose the explicit value function obtained by \cite{Fouque}, and then include the jump term as the benchmark model. Based on it, we compare our approximating value function to the value function obtained in \cite{Kumar}. 

	\begin{theorem}
		If the solution of HJB equation \eqref{HJB} is given by \eqref{U_hat}, then it can be further approximated by 
		\begin{align} \label{Ex1}
			\hat{U}(t,x,y) 
			\approx U_{_T}(x) - \frac{T-t}{2}\bigg[ \lambda^2(y) -\Big[\int_\mathbb{R} \gamma_{1_T}(\zeta) \nu(d\zeta)\Big]^2 \bigg] \frac{\big[U_{_T}'(x) \big]^2}{ U_{_T}''(x)}.
		\end{align}
		In addition, if $ U_{_T}(x) = c x^{-2} $ with some constant $ c $, and for all $\zeta$ and $y$, if $\big| \gamma_{1_T}(\zeta) [ \lambda(y) - \int_\mathbb{R} \gamma_{1_T}(\zeta) \nu(d\zeta)] \big| <3$, then the approximation \eqref{Ex1} converges in the order of $O( n d^n) $, where $ d<1 $, and $n>1$ an integer.

	\end{theorem}

	\begin{proof}
		We denote $ L = \lambda(y) - \int_\mathbb{R} \gamma_{1_T}(\zeta) \nu(d\zeta) $.
		We find in Section \ref{sec3} that the solution of HJB equation \eqref{HJB} is
		\begin{align}
			\label{recall} 
			\hat{U}(t,x,y) =\
			& U_{_T}(x) - \frac{T-t}{2}\bigg[ \lambda^2(y) -\Big[\int_\mathbb{R} \gamma_{1_T}(\zeta) \nu(d\zeta)\Big]^2 \bigg] \frac{\big[U_{_T}'(x) \big]^2}{ U_{_T}''(x)} \nonumber \\
			& + (T-t) \int_\mathbb{R} \Big[ U_{_T}(\hat{\chi}) - U_{_T}(x) + L \frac{\big[U_{_T}'(x)\big]^2}{U_{_T}''(x)} \gamma_{1_T}(\zeta) \Big]\nu(d\zeta),
		\end{align}
		where $\hat{\chi} = x - L \frac{U_{_T}'(x)}{U_{_T}''(x)} \gamma_{1_T}(\zeta)$. We note that $ \int_\mathbb{R} \nu(d\zeta) = \infty $. We expand $ U_{_T}(\hat{\chi}) $ at $ x $ using Taylor series: 
		\begin{align*}
			 U_{_T}(\hat{\chi})
			= \sum_{n=0}^{\infty}
			\frac{U_{_T}^{(n)}(x)}{n!} \Big[ - L \frac{U_{_T}'(x)}{U_{_T}''(x)} \gamma_{1_T}(\zeta) \Big]^n .
		\end{align*}
		We then choose the first order term of Taylor expansion of $U_{_T}(\hat{\chi}) $, that is
		\begin{align} \label{app}
			U_{_T}(\hat{\chi})
			\approx U_{_T}(x) - L \frac{\big[U_{_T}'(x)\big]^2}{U_{_T}''(x)} \gamma_{1_T}(\zeta).
		\end{align}
		Finally we substitute \eqref{app} into \eqref{recall} and obtain the approximation \eqref{Ex1}.

		If $ U_{_T}(x) = c x^{-2} $ with some constant $ c $, then
		\begin{align} \label{Ex13}
			U_{_T}(\hat{\chi})
			= \sum_{n=0}^{\infty}
			[n+1]\Big[- \frac{L}{3} \gamma_{1_T}(\zeta)\Big]^n c x^{-2}.
		\end{align}
		We apply the following ratio test for the convergence of above series \eqref{Ex13}.
		\begin{align*}
			\lim_{n \rightarrow \infty} 
			\Bigg| \frac{[n+2] \big[- \frac{L}{3} \gamma_{1_T}(\zeta)\big]^{n+1} c x^{-2}}
			{[n+1] \big[- \frac{L}{3} \gamma_{1_T}(\zeta)\big]^n c x^{-2}} \Bigg|
			= \Big| - \frac{L}{3} \gamma_{1_T}(\zeta) \Big|.
		\end{align*}
		If $ \big| L \gamma_{1_T}(\zeta) \big| <3 $, then the series \eqref{Ex13} converges.
		With the Lagrange form of remainder, we also observe that
		\begin{align*}
			\lim_{n \rightarrow \infty}
			\frac{c [n+2] \big[ - \frac{L}{3} \gamma_{1_T}(\zeta)x \big]^{n+1} \epsilon^{-n-3}}{n \big[- \frac{L}{3} \gamma_{1_T}(\zeta)x \big]^n \epsilon^{-n}} = c \epsilon^{-3} \Big[- \frac{L}{3} \gamma_{1_T}(\zeta)x \Big],
		\end{align*}
		where $ x < \epsilon < \hat{\chi} $. We note that $ \big| L \gamma_{1_T}(\zeta) \big| <3 \Rightarrow - \frac{L}{3} \gamma_{1_T}(\zeta)x < x $. Then we have $ d = - \frac{L}{3} \gamma_{1_T}(\zeta)x \big/ \epsilon < 1 $.
		Thus, the remainder of series \eqref{Ex13} is of $ O(nd^n) $, where $ d<1 $, and $n>1$ an integer.
	\end{proof}

We note that the value function in \cite{Kumar} is given by
	\begin{align} \label{UR}
		U^R(t,x,y) 
		= U_{_T}(x) - (T-t)\frac{\lambda^2(y)}{2} \frac{\big[U_{_T}'(x) \big]^2}{ U_{_T}''(x)}.
	\end{align}
	Including the jump term to the formula for the value function (as discussed in \cite{Kumar} and \cite{Fouque}), we obtain the benchmark value function
	\begin{align} \label{benchmark1}
		U^B(t, x, y) = -\frac{1}{2x^2}e^{\frac{\gamma}{\gamma+[1-\gamma]\rho^2}[yA(t,T)+B(t,T)]}
		+ \frac{T-t}{2} \Big[\int_\mathbb{R} \gamma_{1_T}(\zeta) \nu(d\zeta)\Big]^2 \frac{\big[U_{_T}'(x) \big]^2}{ U_{_T}''(x)},
	\end{align}
where 
\begin{equation*}
A(t, T)= \frac{(1-e^{-\alpha(T-t)})a_{-}}{1-\frac{a_{-}}{a_{+}}e^{-\alpha(T-t)}}, \quad  B(t,T)= m\left((T-t)a_{-}- \frac{2}{\beta^2} \log \left(\frac{1- \frac{a_{-}}{a_{+}}e^{-\alpha(T-t)} }{1- \frac{a_{-}}{a_{+}}} \right) \right),
\end{equation*}
where the positive and negative roots of $f(r)=0$ are denoted as $a_{+}$ and $a_{-}$, where $$f(r)= \frac{\beta^2}{2}r^2 + \left(\frac{(1-\gamma)\beta \mu \rho- \gamma}{\gamma}\right)r + \frac{(\gamma+(1-\gamma)\rho^2)(1-\gamma)\mu^2}{2 \gamma^2}.$$ 
We use \cite{Kumar} for the values of constants for our subsequent numerical analysis. At this point, we consider specific L\'evy processes as examples. If 
	\begin{align} \label{ldf11}
		\nu(d\zeta) = \kappa \zeta^{-1} e^{-\theta \zeta} d\zeta,
	\end{align}
	with $ \zeta \in (0, \infty) $, where $ t \kappa > 0 $ and $ \theta>0 $ are the shape and the rate of Gamma distribution with probability density function
	$f(\zeta) = \frac{\theta^{t \kappa}}{\Gamma(t \kappa)} \zeta^{t \kappa-1} e^{-\theta \zeta}$, then we let
	\begin{align} \label{ldf12}
		\gamma_{1_T}(\zeta) = \frac{\theta^{t \kappa}}{\kappa \Gamma(t \kappa)} \zeta^{t \kappa}.
	\end{align}
	If
	\begin{align} \label{ldf21}
		\nu(d\zeta) = nt \sqrt{\frac{1}{2\pi \zeta^3}} \exp \Big(-\frac{m^2 \zeta}{2}\Big) d\zeta,
	\end{align}
	with $ \zeta \in (0, \infty) $, where $ m >0 $ and $ n>0 $ are the parameters of inverse Gaussian distribution with probability density function
	\begin{align*} 
		f(\zeta) 
		= nt \exp(mnt) \sqrt{\frac{1}{2\pi \zeta^3}} \exp \Big(-\frac{n^2t^2\zeta^{-1}+m^2 \zeta}{2}\Big),
	\end{align*}
	then we let
	\begin{align} \label{ldf22}
		\gamma_{1_T}(\zeta) = \exp(mnt) \exp \Big(-\frac{n^2t^2\zeta^{-1}}{2}\Big).
	\end{align}
	We substitute \eqref{ldf11} and \eqref{ldf12} into equations \eqref{Ex1} and \eqref{benchmark1}, or substitute \eqref{ldf21} and \eqref{ldf22} into equations \eqref{Ex1} and \eqref{benchmark1},  then obtain the approximating value function
	\begin{align} \label{Ex11}
		\hat{U}(t,x,y) 
		\approx U_{_T}(x) - (T-t)\frac{\big[\lambda^2(y)-1\big]}{2} \frac{\big[U_{_T}'(x) \big]^2}{ U_{_T}''(x)},
	\end{align}
	and the benchmark value function
	\begin{align} \label{benchmark11}
		U^B(t, x, y) = -\frac{1}{2x^2}e^{\frac{\gamma}{\gamma+[1-\gamma]\rho^2}[yA(t,T)+B(t,T)]}
		+ \frac{(T-t)}{2} \frac{\big[U_{_T}'(x) \big]^2}{ U_{_T}''(x)}.
	\end{align}
	Referring to Section 5 in \cite{Kumar}, we let $ \lambda^2(y) = 0.183732 $ and $ U_{_T}(x) = - \frac{1}{2} x^{-2} $. The value function \eqref{UR} (as provided in \cite{Kumar}), approximating value function \eqref{Ex11}, benchmark value function \eqref{benchmark11} and respective errors are summarized in Table 1, and are graphed in Figures 1, 2, 3 and 4, respectively.
	\begin{center}
		Table 1 \\
		\begin{tabular}{c c c c c c c c}
			\hline 
			$t$ & $T$ & $ U^B(t, x, y) $ & $ \hat{U}(t,x,y)  $ & $ U^R(t,x,y)  $ & $ \big|U^B - \hat{U}\big| $ & $ \big|U^B - U^R\big| $  \\
			\hline \\
			1.5 & 2 & $ \approx -\frac{0.568355}{x^2} $ & $ \approx -\frac{0.568022}{x^2} $ & $ \approx -\frac{0.484689}{x^2} $ & $ \approx \frac{0.000333}{x^2} $ & $ \approx \frac{0.083666}{x^2} $ \\
			\\
			1.9 & 2 & $ \approx -\frac{0.513619}{x^2} $ & $ \approx -\frac{0.513605}{x^2} $ & $ \approx -\frac{0.496938}{x^2} $ & $ \approx \frac{0.000014}{x^2} $ & $ \approx \frac{0.016681}{x^2} $ \\
			\\
			\hline 
		\end{tabular}
	\end{center}
	\begin{figure}[htbp]
		\centering
		\begin{minipage}{\textwidth}
			\centering
			\includegraphics[width=0.7\textwidth]{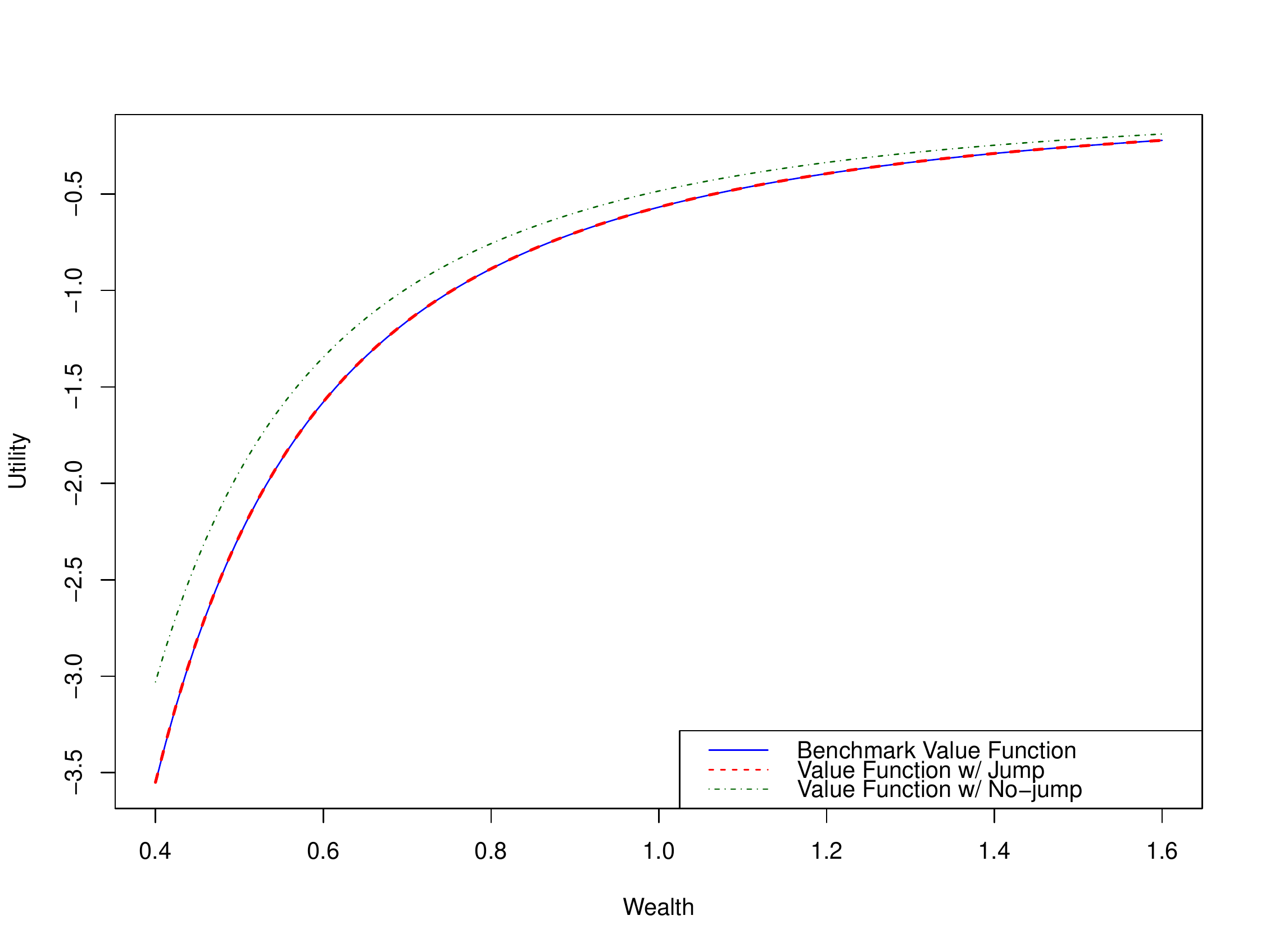}
			\caption{($t = 1.5$, $T = 2$) The benchmark value function is plotted against the approximating value function with jump and
				the value function in \cite{Kumar}. It is difficult to distinguish between the benchmark value function and the approximating value function with jump.}
		\end{minipage}
		\begin{minipage}{\textwidth}
			\centering
			\includegraphics[width=0.7\textwidth]{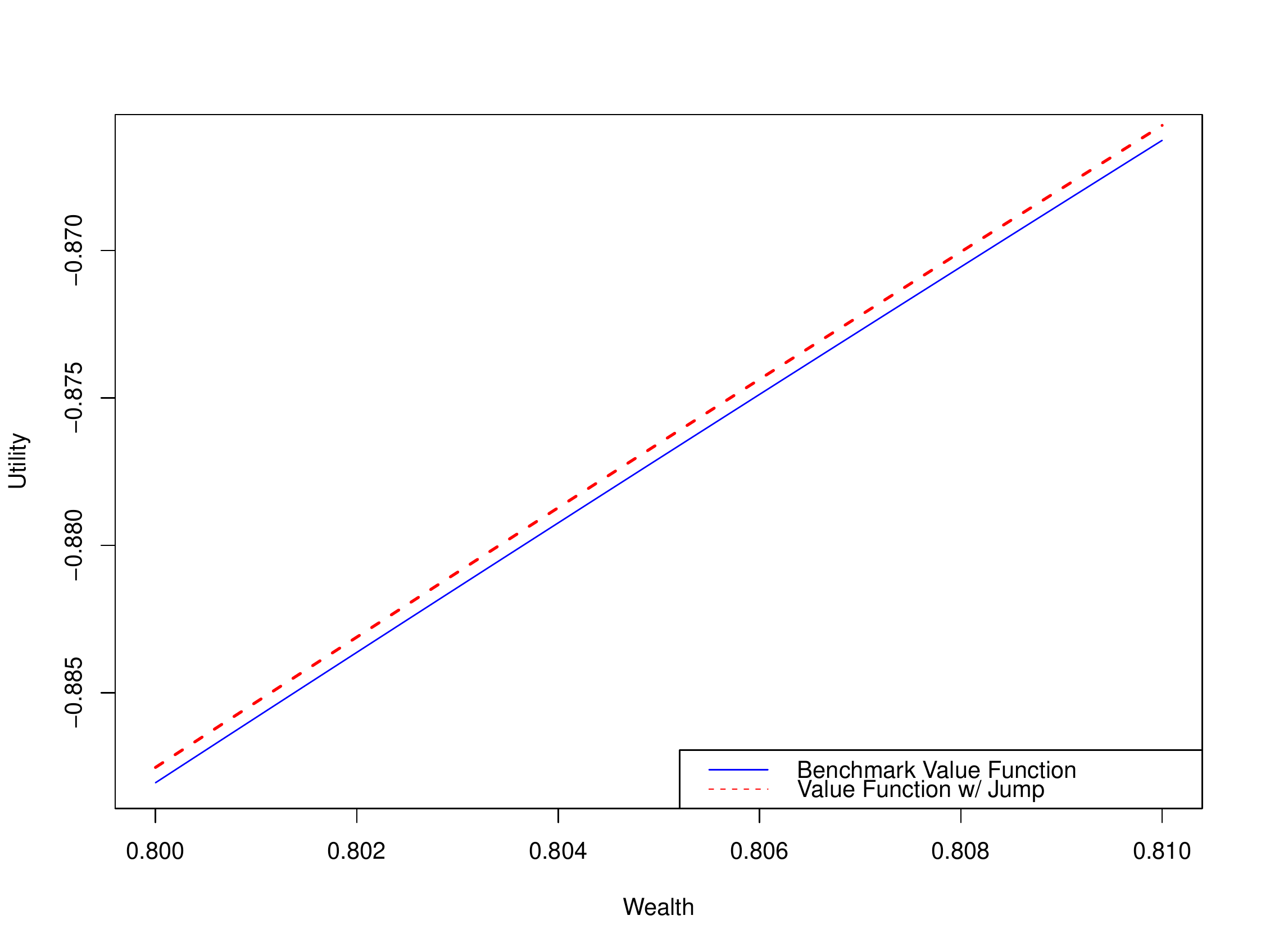}
			\caption{($t = 1.5$, $T = 2$) When Figure 1 is
				zoomed in over a shorter wealth interval, difference between the approximating value function with jump and the benchmark value function is more apparent.}
		\end{minipage}
	\end{figure}
	\begin{figure}[htbp]
		\centering
		\begin{minipage}{\textwidth}
			\centering
			\includegraphics[width=0.7\textwidth]{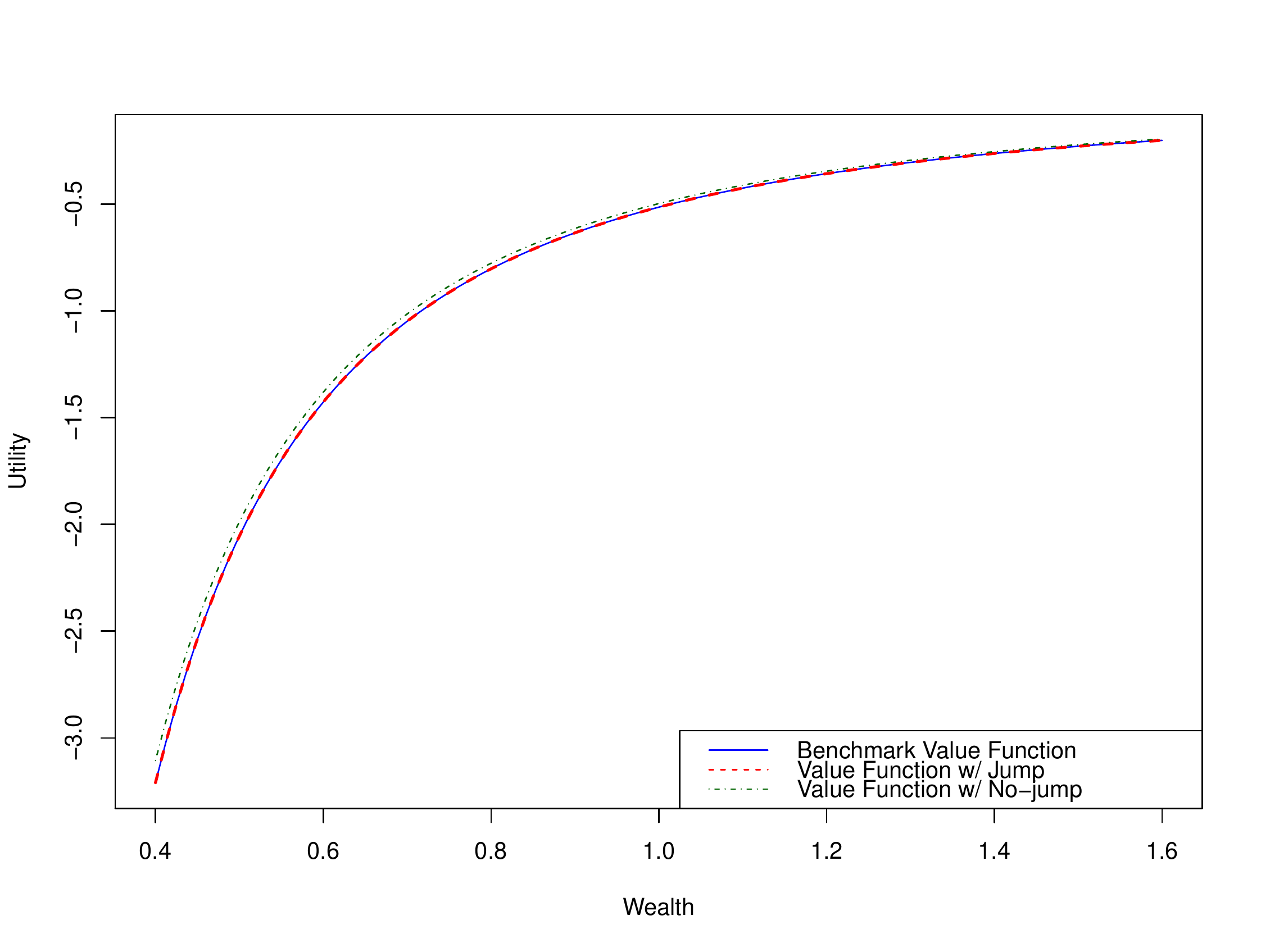}
			\caption{($t = 1.9$, $T = 2$) The benchmark value function is plotted against the approximating value function with jump and
				the value function in \cite{Kumar}. It is difficult to distinguish between the benchmark value function and the approximating value function with jump.}
		\end{minipage}
		\begin{minipage}{\textwidth}
			\centering
			\includegraphics[width=0.7\textwidth]{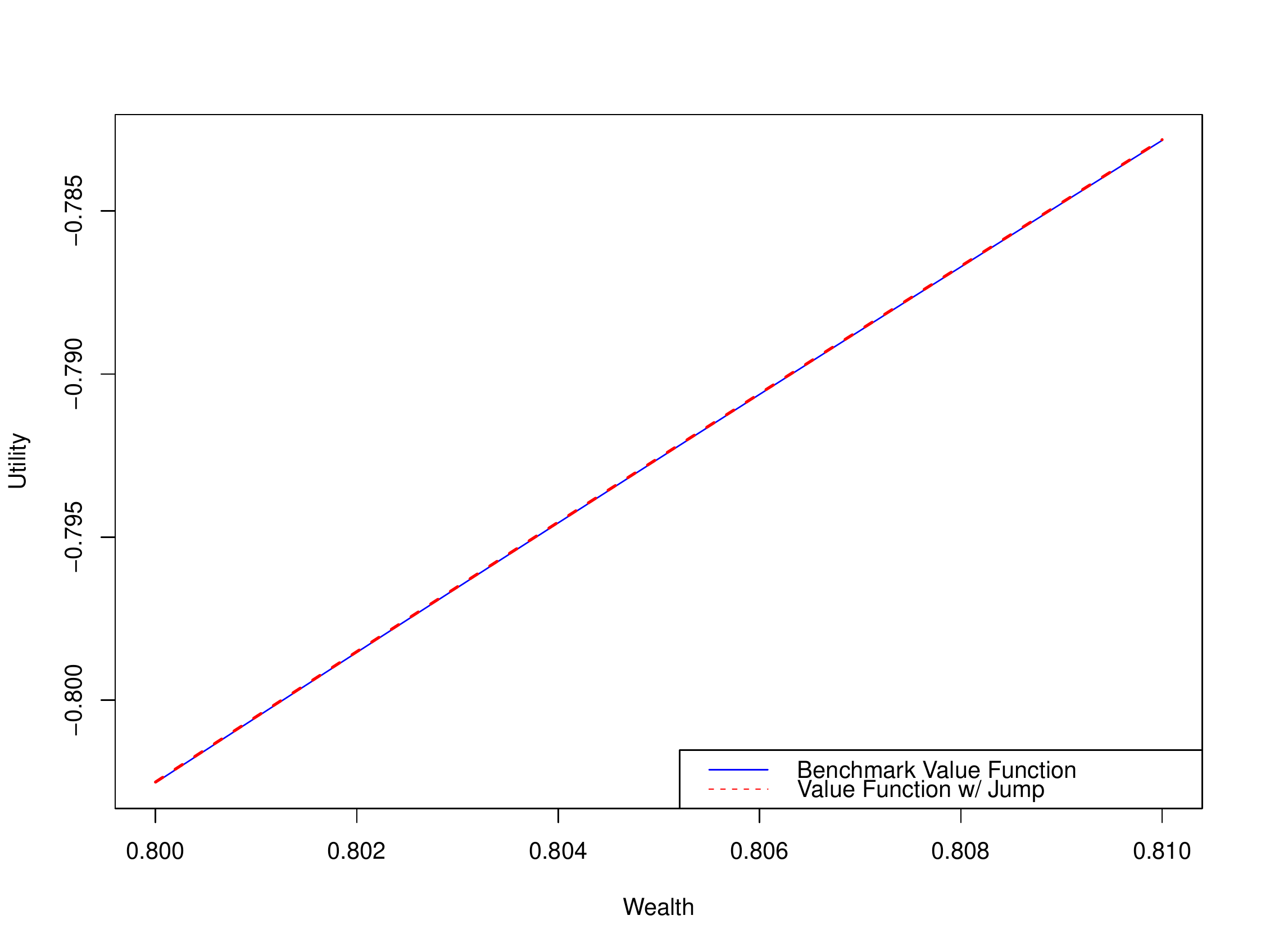}
			\caption{($t = 1.9$, $T = 2$) When time interval $ T-t $ is shortened from $ 0.5 $ to $ 0.1 $, the approximating value function with jump is much closer to the benchmark value function.}
		\end{minipage}
	\end{figure}

\section{Approximating portfolio}
\label{sec5}

In this section, we generate a close-to-optimal portfolio near the time to horizon $ T-t $ by the first-order approximation of utility function. The generated portfolio yields an expected utility function close to the maximum expected utility. We also control the error by the square of the time to the horizon $ (T-t)^2 $.

	\begin{theorem}
		Let $ \tau \in [t, T] $. If the total wealth $ \hat{X}_\tau $ is given by
		\begin{align*}
			d\hat{X}_\tau = \sigma(Y_\tau) \tilde{\pi}(\tau,\hat{X}_\tau,Y_\tau) \Big[ \lambda(Y_\tau)d\tau + dW_\tau^{(1)} + \int_\mathbb{R} \gamma_1(\tau,\zeta) \tilde{N}(d\tau,d\zeta) \Big], 
		\end{align*}
		where
		\begin{align*}
			\tilde{\pi}(\tau,\hat{X}_\tau,Y_\tau)
			= \
			& \frac{ -\lambda(Y_\tau)\hat{U}_{\hat{X}_\tau}(\tau,\hat{X}_\tau,Y_\tau) - \rho a(Y_\tau) \hat{U}_{\hat{X}_\tau Y_\tau}(\tau,\hat{X}_\tau,Y_\tau)}
			{\sigma(Y_\tau)\hat{U}_{\hat{X}_\tau \hat{X}_\tau}(\tau,\hat{X}_\tau,Y_\tau)} \\
			& + \frac{\hat{U}_{\hat{X}_\tau}(\tau,\hat{X}_\tau,Y_\tau) }
			{\sigma(Y_\tau)\hat{U}_{\hat{X}_\tau \hat{X}_\tau}(\tau,\hat{X}_\tau,Y_\tau)} \int_\mathbb{R}  \gamma_1(\tau, \zeta) \nu(d\zeta),
		\end{align*}
		with utility function $ \hat{U} $ given by \eqref{U_hat}. Then, there exists constant $ c > 0 $ and $ 0 < \varepsilon < \min\{1, T\} $ such that
		\begin{multline*}
			\big| J(t,x,y) - \mathbb{E}\big( U_{_T}(\hat{X}_{_T}) \big| \hat{X}_t = x, Y_t = y \big) \big| \leq c (T-t)^2f(x) \\ \text{for }  (t,x,y) \in (T-\varepsilon,T)\times \mathbb{R}^+ \times \mathbb{R},
		\end{multline*}
		where $ J(t,x,y) $ is value function defined by \eqref{J}; $ f(x) = 1 $ under Case 1 of Assumption \ref{Assume_UT} and $ f(x) = x^{1-\alpha} + x^{1-\beta} $ under Case 2 of Assumption \ref{Assume_UT}; the constants $ c $ and $\varepsilon$ are independent of $ t $, $ x $ and $ y $.
	\end{theorem}
	\begin{proof}
		We denote $ \lambda = \lambda(Y_\tau) $, $ \sigma = \sigma(Y_\tau) $, $ b = b(Y_\tau) $, $ a = a(Y_\tau) $, $ \tilde{\pi} = \tilde{\pi}(\tau,\hat{X}_\tau,Y_\tau) $ and $ \hat{U} = \hat{U}(\tau, \hat{X}_\tau, Y_\tau) $. We let $ \{t_n \}_{n=1}^\infty \subset [t, T] $ be a sequence of stopping times such that $ t_n \leq t_{n+1} $ and $ t_n \rightarrow T $. 
		
		We apply the 2-dimensional It\^{o}'s formula to $ \hat{U} $ and obtain
		\begin{align}
			& \hat{U}(t_n, \hat{X}_{t_n}, Y_{t_n}) - \hat{U}(t,x,y)  \nonumber\\
			& = \int_t^{t_n} \Big[ \partial_\tau \hat{U} + \lambda\sigma\tilde{\pi}\hat{U}_{\hat{X}_\tau} + b\hat{U}_{Y_\tau}
			+ \frac{1}{2}\sigma^2\tilde{\pi}^2\hat{U}_{\hat{X}_\tau \hat{X}_\tau} + \rho a\sigma\tilde{\pi} \hat{U}_{\hat{X}_\tau Y_\tau} +\frac{1}{2}a^2\hat{U}_{Y_\tau Y_\tau} \Big] d\tau \label{HJB_5} \\
			& \quad
			+ \int_t^{t_n} \Big[ \sigma\tilde{\pi}\hat{U}_{\hat{X}_\tau} + \rho a \hat{U}_{Y_\tau} \Big] dW^{(1)}_\tau
			+ \int_t^{t_n} \Big[ a\sqrt{1-\rho^2}\hat{U}_{Y_\tau} \Big] dW^{(2)}_\tau \label{mart_5} \\
			& \quad
			+ \int_t^{t_n} \int_\mathbb{R} \Big[ \hat{U}\big(\tau, \hat{X}_\tau+ \sigma\tilde{\pi}\gamma_1(\tau,\zeta), Y_\tau+\gamma_2(\tau,\zeta)\big) - \hat{U} - \sigma\tilde{\pi}\gamma_1(\tau,\zeta)\hat{U}_{\hat{X}_\tau} \nonumber\\
			& \qquad \qquad \qquad
			- \gamma_2(\tau,\zeta)\hat{U}_{Y_\tau} \Big] \nu(d\zeta) d\tau \label{HJB_6}\\
			& \quad
			+ \int_t^{t_n} \int_\mathbb{R} \Big[ \hat{U}\big(\tau, \hat{X}_\tau+ \sigma\tilde{\pi}\gamma_1(\tau,\zeta), Y_\tau+\gamma_2(\tau,\zeta)\big) - \hat{U} \Big]\tilde{N}(d\tau,d\zeta), \label{mart_6} 
		\end{align}
		where the stochastic integrals \eqref{mart_5} and \eqref{mart_6} are martingales. We note that the integrand of \eqref{HJB_5} $ + $ \eqref{HJB_6} is exactly the $ \hat{U}_\tau + \mathcal{H}(\hat{U}) $. 
		Because $ U^{(1)} $ is solved by ignoring the terms of order $ O(1) $, therefore $ \big| \hat{U}_\tau + \mathcal{H}(\hat{U}) \big| = O(T-\tau)O\big(f(\hat{X}_\tau)\big) $ under Assumption \ref{Assume_UT}.
		By the martingale property that conditional expectation of martingale is zero, we have
		\begin{multline*}
			\mathbb{E}\big( \hat{U}(t_n, \hat{X}_{t_n}, Y_{t_n}) \big| \hat{X}_t = x, Y_t = y \big) - \hat{U}(t,x,y) \\
			= \int_t^{t_n} \mathbb{E} \Big( O(T-\tau)O\big(f(\hat{X}_\tau)\big) \Big| \hat{X}_t = x, Y_t = y \Big) d\tau.
		\end{multline*}
		Then referring to inequality (4.3) in \cite{Kumar}, we have
		\begin{align} \label{4.3}
			\Big| \mathbb{E}\big( \hat{U}(t_n, \hat{X}_{t_n}, Y_{t_n}) \big| \hat{X}_t = x, Y_t = y \big) - \hat{U}(t,x,y) \Big|  \leq c_1 (T-t)^2 f(x),
		\end{align}
		with some constant $ c_1>0 $.
		
		We apply equation \eqref{U_hat} and triangle inequality to $ \hat{U}(t_n, \hat{X}_{t_n}, Y_{t_n}) $ and obtain
		\begin{align*}
			\big| \hat{U}(t_n, \hat{X}_{t_n}, Y_{t_n}) \big|
			& = \big| U_{_T}(\hat{X}_{t_n}) + (T-t_n)U^{(1)}(\hat{X}_{t_n},Y_{t_n}) \big| \\
			& \leq \big| U_{_T}(\hat{X}_{t_n}) \big| + \big| T U^{(1)}(\hat{X}_{t_n},Y_{t_n}) \big| \\
			& \leq c_2 g(\hat{X}_{t_n}),
		\end{align*}
		with some constant $ c_2 $. We recall equation \eqref{sim3} showing $ U^{(1)} \sim f(\hat{X}_{t_n}) $. Then we have $ g(\hat{X}_{t_n}) = U_{_T}(\hat{X}_{t_n}) + f(\hat{X}_{t_n}) = \ln(\hat{X}_{t_n}) +1 $ under Case 1 of Assumption \ref{Assume_UT}.
		In addition, we recall equation \eqref{sim12} showing $ U_{_T}(\hat{X}_{t_n}) \sim f(\hat{X}_{t_n}) $ under Case 2 of Assumption \ref{Assume_UT}. Thus we have $ g(\hat{X}_{t_n}) = f(\hat{X}_{t_n}) = \hat{X}_{t_n}^{1-\alpha} + \hat{X}_{t_n}^{1-\beta} $ under this case. 
		
		With substituting $ \hat{X}_{t_n} $ into $ X_{t_n} $, the Lemma \ref{A.1} in the Appendix \ref{apendix} also proves that $ \big\{ g(\hat{X}_{t_n}) \big\}_{n=1}^\infty $ is uniformly bounded by an integrable random variable.  We already have $ \big| \hat{U}(t_n, \hat{X}_{t_n}, Y_{t_n}) \big| \leq c_2 g(\hat{X}_{t_n}) $. And we observe that
		\begin{align*}
			t_n \rightarrow T\ \Rightarrow\ \hat{U}(t_n, \hat{X}_{t_n}, Y_{t_n}) \rightarrow \hat{U}(T, \hat{X}_{_T}, Y_{_T}) = U_{_T}(\hat{X}_{_T}).
		\end{align*}
		By Dominated Convergence Theorem, we have
		\begin{align} \label{lim2}
			\lim_{n \rightarrow \infty} \mathbb{E}\big( \hat{U}(t_n, \hat{X}_{t_n}, Y_{t_n}) \big| \hat{X}_t = x, Y_t = y \big) = \mathbb{E}\big(U_{_T}(\hat{X}_{_T}) \big| \hat{X}_t = x, Y_t = y \big).
		\end{align}
		Combing \eqref{4.3} with \eqref{lim2}, we have
		\begin{align} \label{4.5}
			\Big| \mathbb{E}\big(U_{_T}(\hat{X}_{_T}) \big| \hat{X}_t = x, Y_t = y \big) - \hat{U}(t,x,y) \Big|  \leq c_1 (T-t)^2 f(x).
		\end{align}
		By \eqref{ineq}, \eqref{4.5} and triangle inequality, there exists constant $ c > 0 $ and $ 0 < \varepsilon < \min\{1, T\} $ such that
		\begin{align*}
			& \big| J(t,x,y) - \mathbb{E}\big( U_{_T}(\hat{X}_{_T}) \big| \hat{X}_t = x, Y_t = y \big) \big| \\
			& \leq \big| J(t,x,y) - \hat{U}(t,x,y) \big| + \big| \hat{U}(t,x,y) - \mathbb{E}\big( U_{_T}(\hat{X}_{_T}) \big| \hat{X}_t = x, Y_t = y \big) \big| \\
			& \leq c(T-t)^2 f(x).
		\end{align*}
		for $ (t,x,y) \in (T-\varepsilon,T)\times \mathbb{R}^+ \times \mathbb{R} $.
	\end{proof}

\section{Portfolio optimization on a finite time horizon}
\label{sec6}

In this section, we approximate the value function for all times $ t \in [0,T] $. Using this approximation with optimal portfolio \eqref{pi},  we generate a close-to-optimal portfolio on $ [0, T] $. To start, we partition the interval $ [0, T] $ into $ n $ subintervals: $ \{ 0 = t_0 < t_1 < \cdots < t_{n-1} < t_n = T \} $. For $ t_k \leq t \leq t_{k+1} $, $ k = 0, \cdots, n-1 $, the approximation scheme is given by 
	\begin{align*}
		\hat{U}(t, x, y) := \hat{U}(t_{k+1}, x, y) + (t_{k+1} - t)\mathcal{H}\big(\hat{U}(t_{k+1}, x, y)\big),
	\end{align*}
	where
	\begin{align*}
		& \mathcal{H}\big(\hat{U}(t_{k+1}, x, y)\big) \\
		& = \lambda(y)\sigma(y)\tilde{\pi}(t_{k+1}, x, y)\hat{U}_x(t_{k+1}, x, y) + b(y)\hat{U}_y(t_{k+1}, x, y) \\ 
		& \quad
		+ \frac{1}{2}\sigma^2(y)\tilde{\pi}^2(t_{k+1}, x, y)\hat{U}_{xx}(t_{k+1}, x, y) +\frac{1}{2}a^2(y)\hat{U}_{yy}(t_{k+1}, x, y) \\
		& \quad
		+ \rho a(y)\sigma(y)\tilde{\pi}(t_{k+1}, x, y) \hat{U}_{xy}(t_{k+1}, x, y) \\
		& \quad
		+ \int_\mathbb{R} \Big[ \hat{U} \big(t_{k+1}, x+ \sigma(y)\tilde{\pi}(t_{k+1},x,y)\gamma_1(t_{k+1},\zeta), y+\gamma_2(t_{k+1},\zeta)\big) - \hat{U}(t_{k+1}, x, y) \\
		& \quad \quad \quad \quad
		- \sigma(y)\tilde{\pi}(t_{k+1}, x, y)\gamma_1(t_{k+1},\zeta) \hat{U}_x(t_{k+1}, x, y)
		- \gamma_2(t_{k+1},\zeta)\hat{U}_y(t_{k+1}, x, y)
		 \Big] \nu(d\zeta),
	\end{align*}
	with
	\begin{align*}
		\tilde{\pi}(t_{k+1},x,y)
		= \
		& \frac{-\lambda(y)\hat{U}_x(t_{k+1},x,y) - \rho a(y) \hat{U}_{xy}(t_{k+1},x,y)}{\sigma(y)\hat{U}_{xx}(t_{k+1},x,y)} \\
		& + \frac{\hat{U}_x(t_{k+1},x,y)}{\sigma(y)\hat{U}_{xx}(t_{k+1},x,y)} 
		\int_\mathbb{R}  \gamma_1(t_{k+1}, \zeta) \nu(d\zeta).
	\end{align*}
	The close-to-optimal portfolio is then given by
	\begin{align*}
		\tilde{\pi}(t,\hat{X}_t,Y_t)
		= \
		& \frac{ -\lambda(Y_t)\hat{U}_{\hat{X}_t}(t,\hat{X}_t,Y_t) - \rho a(Y_t) \hat{U}_{\hat{X}_t Y_t}(t,\hat{X}_t,Y_t)}
		{\sigma(Y_t)\hat{U}_{\hat{X}_t \hat{X}_t}(t,\hat{X}_t,Y_t)} \nonumber \\
		& + \frac{\hat{U}_{\hat{X}_t}(t,\hat{X}_t,Y_t) }
		{\sigma(Y_t)\hat{U}_{\hat{X}_t \hat{X}_t}(t,\hat{X}_t,Y_t)} \int_\mathbb{R}  \gamma_1(t, \zeta) \nu(d\zeta).
	\end{align*}
	
	We note that $ \int_\mathbb{R} \nu(d\zeta) = \infty $. To calculate $ \int_\mathbb{R} - \hat{U}(t_{k+1}, x, y) \nu(d\zeta) $, we consider to expand $ \hat{U} \big(t_{k+1}, x+ \sigma(y)\tilde{\pi}(t_{k+1},x,y)\gamma_1(t_{k+1},\zeta), y+\gamma_2(t_{k+1},\zeta)\big) $ at $ (x, y) $ using the first order of Taylor series:
	\begin{align*}
		& \hat{U} \big(t_{k+1}, x+ \sigma(y)\tilde{\pi}(t_{k+1},x,y)\gamma_1(t_{k+1},\zeta), y+\gamma_2(t_{k+1},\zeta)\big) \\
		& = \hat{U}(t_{k+1}, x, y) +  \sigma(y)\tilde{\pi}(t_{k+1}, x, y)\gamma_1(t_{k+1},\zeta)\hat{U}_x(t_{k+1}, x, y) + \gamma_2(t_{k+1},\zeta)\hat{U}_y(t_{k+1}, x, y) \\
		& \quad 
		+ O\big(\sigma^2(y)\tilde{\pi}^2(t_{k+1}, x, y)\gamma_1^2(t_{k+1},\zeta) + \gamma_2^2(t_{k+1},\zeta)\big).
	\end{align*}
	In addition, if $ \lambda(y) = \lambda$, is independent of $ y $, then $ \hat{U}(t, x, y) = \hat{U}(t, x) $ by terminal condition $ \hat{U}(T, x, y) = U_{_T}(x) $. Thus the approximation scheme is given by
	\begin{align*}
		\hat{U}(t, x) \approx \hat{U}(t_{k+1}, x) + 
		\frac{t_{k+1} - t}{2} 
		\bigg[ - \lambda^2 + \Big[ \int_\mathbb{R}  \gamma_1(t_{k+1}, \zeta) \nu(d\zeta)\Big]^2 \bigg]
		\frac{\big[ \hat{U}_x(t_{k+1},x) \big]^2}{\hat{U}_{xx}(t_{k+1},x)}.
	\end{align*}
	
	We let $ \gamma_1(t_{k+1}, \zeta) $ be either \eqref{ldf12} or \eqref{ldf22}, and obtain the approximation scheme
	\begin{align*}
		\hat{U}(t, x) \approx \hat{U}(t_{k+1}, x) + 
		(t_{k+1} - t)
		\frac{\big[- \lambda^2 + 1\big]}{2}
		\frac{\big[ \hat{U}_x(t_{k+1},x) \big]^2}{\hat{U}_{xx}(t_{k+1},x)}.
	\end{align*}
	We let $ \gamma_1(t, \zeta) $ be either \eqref{ldf12} or \eqref{ldf22}, and obtain the close-to-optimal portfolio
	\begin{align*}
		\tilde{\pi}(t,x,y)
		= \
		& \frac{ [-\lambda +1] \hat{U}_x(t,x)}
		{\sigma(y)\hat{U}_{xx}(t,x)}.
	\end{align*}
	Referring to Section 5 in \cite{Kumar}, we let $ T=2 $, $ \lambda^2 = 0.183732 $, $ U_{_T}(x) = - \frac{1}{2} x^{-2} $ and $ \sigma(y) = y^{-1/2} $ with $ y = 27.9345 $. The close-to-optimal portfolio is calculated as $ \tilde{\pi} = -1.00605x $ at time $ t=0 $.

\section{Conclusion}
\label{sec7}

In this paper, we consider the finite horizon portfolio optimization in a L\'evy-process-setting where the stochastic volatility portfolio process is driven by a standard Brownian motion and a jump term. The value function is approximated using the polynomial expansion method with respect to time to the horizon ($T-t$). We obtain an approximate solution for the value function and optimal investment strategy. It is shown that the first-order approximations of the value function and optimal investment strategy perform better than the existing models such as \cite{Kumar}. It is shown that the first-order term in the value function approximation can always be expressed in terms of the zeroth-order term and its derivative. For certain utility functions, it is shown that the convergence is linear. For certain utility functions, it is shown that the remainder of approximation is of the order of $O( nd^n) $, where $ d<1 $, and $n>1$ an integer.

Based on our approximation, we also generate a close-to-optimal portfolio near the time to horizon $ (T - t) $. We provide an approximation scheme to the value function for all times $ t \in [0, T] $ and generate the close-to-optimal portfolio on $ [0, T] $. The accuracy of such approximation can be accomplished by a similar procedure used in Section \ref{sec3} and will be rigorously proved in a sequel of this work.

\appendix

\section{Appendix}
\label{apendix}

	\begin{lemma} \label{A.0}
		Known that $ U^{(2)}(x,y) $ is given by \eqref{U^2}. If take $ U_{_T}(\hat{\chi}) - U_{_T}(x) $ as one term and enumerate all the terms of $ U^{(2)}(x,y) $ as $ u_1^{(2)}(x,y) , \cdots , u_l^{(2)}(x,y) $, i.e.\ $ U^{(2)}(x,y) = \sum_{i=1}^l u_i^{(2)}(x,y) $, then
		\begin{align*}
			u_i^{(2)}(x,y) \sim f(x),
		\end{align*}
		for $ 1\leq i \leq l $, where $ f(x) = 1 $ under Case 1 of Assumption \ref{Assume_UT} and $ f(x) = x^{1-\alpha} + x^{1-\beta} $ under Case 2 of Assumption \ref{Assume_UT}.
	\end{lemma}
	\begin{proof}
		We denote $ L = \lambda(y) - \int_\mathbb{R} \gamma_{1_T}(\zeta) \nu(d\zeta) $. We let $ f(x) = 1 $ under Case 1 of Assumption \ref{Assume_UT}. Then
		\begin{align*}
			& \lim_{x \rightarrow \infty} \frac{\big[U_{_T}'(x)\big]^2 \big/ U_{_T}''(x)}{f(x)}
			= \lim_{x \rightarrow \infty} \frac{[x^{-1}]^2}{-x^{-2}} = -1,\\
			& \lim_{x \rightarrow \infty} \frac{U_{_T} (\hat{\chi}) - U_{_T}(x)}{f(x)} 
			= \lim_{x \rightarrow \infty} \ln \Big[ x + L \frac{x^{-1}}{x^{-2}} \gamma_{1_T}(\zeta) \Big] - \ln x = \ln \Big[ 1 + L \gamma_{1_T}(\zeta)  \Big], \\
			& \lim_{x \rightarrow \infty} \frac{f(\hat{\chi})}{f(x)} = 1.
		\end{align*}
		We let $ f(x) = x^{1-\alpha} + x^{1-\beta} $ under Case 2 of Assumption \ref{Assume_UT}. Then
		\begin{align}
			& \lim_{x \rightarrow \infty} \frac{\big[U_{_T}'(x)\big]^2 \big/ U_{_T}''(x)}{f(x)}
			= \lim_{x \rightarrow \infty} \frac{[c_1x^{-\alpha}+c_2x^{-\beta}]^2}{[-c_1\alpha x^{-\alpha} - c_2\beta x^{-\beta}][x^{-\alpha} + x^{-\beta}]} 
			= -\frac{c_1}{\alpha}\ \text{or} -\frac{c_2}{\beta},\nonumber \\
			& \lim_{x \rightarrow \infty} \frac{U_{_T}(x)}{f(x)} 
			= \lim_{x \rightarrow \infty} \frac{\frac{c_1}{1-\alpha}x^{1-\alpha}+\frac{c_2}{1-\beta}x^{1-\beta}}{x^{1-\alpha} + x^{1-\beta}}
			= \frac{c_1}{1-\alpha} \text{ or } \frac{c_2}{1-\beta},\label{sim12} \\
			& \lim_{x \rightarrow \infty} \frac{f(\hat{\chi})}{f(x)}
			= \lim_{x \rightarrow \infty}
			\frac{\hat{\chi}^{1-\alpha} + \hat{\chi}^{1-\beta}}{x^{1-\alpha} + x^{1-\beta}}
			= A^{1-\alpha} \text{ or }
			B^{1-\beta}, \nonumber
		\end{align}
		where
		\begin{align*}
			\hat{\chi} = x + L \frac{c_1x^{-\alpha}+c_2x^{-\beta}}{c_1\alpha x^{-\alpha-1} + c_2\beta x^{-\beta-1}} \gamma_{1_T},
			A = 1+\frac{L}{\alpha}\gamma_{1_T}, \text{and}\ B = 1+\frac{L}{\beta}\gamma_{1_T}.
		\end{align*}
		Thus under both Cases of Assumption \ref{Assume_UT}, we find $ f(\hat{\chi}) \sim f(x) $ and 
		\begin{align} \label{sim1}
			\frac{\big[U_{_T}'(x)\big]^2}{U_{_T}''(x)} \sim f(x).
		\end{align}
		Under Case 2 of Assumption \ref{Assume_UT}, there exists $ U_{_T}(x) \sim f(x) $ which imples $ U_{_T}(\hat{\chi}) \sim f(\hat{\chi}) $. We then combine $ U_{_T}(\hat{\chi}) \sim f(\hat{\chi}) $ with $ f(\hat{\chi}) \sim f(x) $ and obtain $ U_{_T}(\hat{\chi}) \sim f(x) $. Thus we have $ U_{_T} (\hat{\chi}) - U_{_T}(x) \sim f(x) $ under this case. In addition, there exists $ U_{_T} (\hat{\chi}) - U_{_T}(x) \sim f(x) $ under Case 1 of Assumption \ref{Assume_UT}. Thus under both Cases of Assumption \ref{Assume_UT}, we find
		\begin{align}\label{sim2}
			U_{_T} (\hat{\chi}) - U_{_T}(x) \sim f(x).
		\end{align}
		We apply \eqref{sim1} into equations \eqref{Uy} and \eqref{Uyy} and find
		\begin{align}\label{sim4}
			U^{(1)}_y \sim f(x) \text{ and } U^{(1)}_{yy} \sim f(x).
		\end{align}
		Also we apply \eqref{sim1} and \eqref{sim2} into equation \eqref{U^1} and find
		\begin{align}\label{sim3}
			U^{(1)}(x,y) \sim f(x),
		\end{align}
		which implies $ U^{(1)}(\hat{\chi},\psi) \sim f(\hat{\chi}) $. Combining $ U^{(1)}(\hat{\chi},\psi) \sim f(\hat{\chi}) $ with $ f(\hat{\chi}) \sim f(x) $, we find
		\begin{align}\label{sim5}
			U^{(1)}(\hat{\chi},\psi) \sim f(x).
		\end{align}
		
		Under Case 1 of Assumption \ref{Assume_UT}, we have
		\begin{align*}
			\lim_{x \rightarrow \infty} \frac{U_{_T}'(\hat{\chi})}{U_{_T}'(x)}
			= \lim_{x \rightarrow \infty}
			\frac{\hat{\chi}^{-1}}{x^{-1}} = A^{-1} \text{ or } B^{-1},
			\lim_{x \rightarrow \infty} \frac{U_{_T}''(\hat{\chi})}{U_{_T}''(x)}
			= \lim_{x \rightarrow \infty}
			\frac{\hat{\chi}^{-2}}{x^{-2}}
			= A^{-2} \text{ or } B^{-2}.
		\end{align*}
		Under Case 2 of Assumption \ref{Assume_UT}, we have
		\begin{align*}
			\lim_{x \rightarrow \infty} \frac{U_{_T}'(\hat{\chi})}{U_{_T}'(x)}
			& = \lim_{x \rightarrow \infty}
			\frac{c_1\hat{\chi}^{-\alpha}+c_2\hat{\chi}^{-\beta}}{c_1x^{-\alpha}+c_2x^{-\beta}}
			= A^{-\alpha} \text{ or }
			B^{-\beta}, \\
			\lim_{x \rightarrow \infty} \frac{U_{_T}''(\hat{\chi})}{U_{_T}''(x)}
			& = \lim_{x \rightarrow \infty}
			\frac{c_1 \alpha \hat{\chi}^{-\alpha-1} + c_2 \beta \hat{\chi}^{-\beta-1}}{c_1 \alpha x^{-\alpha-1} + c_2 \beta x^{-\beta-1}} 
			= A^{-\alpha-1} \text{ or }
			B^{-\beta-1}.
		\end{align*}
		Thus under both Cases of Assumption \ref{Assume_UT}, we find
		\begin{align}
			U_{_T}'(\hat{\chi}) & \sim U_{_T}'(x), \label{sim10}\\
			U_{_T}''(\hat{\chi}) & \sim U_{_T}''(x). \label{sim11}
		\end{align}
		Form \eqref{key1'}, \eqref{key2'} and \eqref{key2''}, we observe that 
		\begin{align}
			\Big(\frac{U_{_T}'(x)}{U_{_T}''(x)}\Big)' & \sim 1, \label{sim(1)} \\
			\Big(\frac{\big[U_{_T}'(x) \big]^2}{U_{_T}''(x)}\Big)' & \sim U_{_T}'(x), \label{sim(3)} \\
			\Big(\frac{\big[U_{_T}'(x) \big]^2}{U_{_T}''(x)}\Big)'' & \sim U_{_T}''(x). \label{sim(4)}
		\end{align}
		We compare all the terms of \eqref{key1''} with the last 3 terms of \eqref{key2''}. Then by \eqref{sim10} and \eqref{sim(4)}, we have
		\begin{align}\label{sim(2)}
			U_{_T}'(\hat{\chi}) \Big(\frac{U_{_T}'(x)}{U_{_T}''(x)}\Big)'' \sim U_{_T}''(x).
		\end{align}
		We apply \eqref{sim11},  \eqref{sim(1)}, \eqref{sim(4)} and \eqref{sim(2)} into equation \eqref{U1xx} and find $ U^{(1)}_{xx} \sim U_{_T}''(x) $, that is
		\begin{align}\label{sim7}
			\frac{U_{xx}^{(1)}}{U_{_T}''(x)} \sim 1.
		\end{align}
		Similarly, we apply \eqref{sim10},  \eqref{sim(1)} and \eqref{sim(3)} into equation \eqref{U1x} and find 
		\begin{align}\label{sim13}
			U_x^{(1)} \sim U_{_T}'(x),
		\end{align}
		then apply \eqref{sim(3)} into equation \eqref{U_xy} and find 
		\begin{align}\label{sim14}
			U_{xy}^{(1)} \sim U_{_T}'(x).
		\end{align}
		Finally we apply \eqref{sim13} and \eqref{sim14} into equation \eqref{sim1} and find
		\begin{align}\label{sim9}
			\frac{U_{_T}'(x)}{U_{_T}''(x)} U_x^{(1)} \sim f(x) \text{ and } \frac{U_{_T}'(x)}{U_{_T}''(x)} U_{xy}^{(1)} \sim f(x).
		\end{align}
		
		In summary, we apply \eqref{sim1}, \eqref{sim2}, \eqref{sim4}, \eqref{sim3}, \eqref{sim5}, \eqref{sim7} and \eqref{sim9} into \eqref{U^2} and prove $ u_i^{(2)}(x,y) \sim f(x) $.
	\end{proof}
	\begin{lemma} \label{A.1}
		Let $ x = X_t $ be the process of total wealth defined by \eqref{Xt}, under any admissible portfolio. Let $ g(x) = \ln(x) +1 $ under Case 1 of Assumption \ref{Assume_UT}, and $ g(x) = x^{1-\alpha} + x^{1-\beta} $, $ \alpha, \beta > 0 $ and $ \alpha, \beta \neq 1 $, under Case 2 of Assumption \ref{Assume_UT}. Then $ \big\{ g(X_{t_n}) \big\}_{n=1}^\infty $ is uniformly bounded by an integrable random variable, where $ \{t_n \}_{n=1}^\infty \subset [t, T] $ is a sequence of stopping times s.t. $ t_n \leq t_{n+1} $ and $ t_n \rightarrow T $.
	\end{lemma}
	\begin{proof}
		We denote $ \sigma_t = \sigma(Y_t) $, $ \pi_t = \pi(t, X_t, Y_t) $ and $ \lambda_t = \lambda(Y_t) $, and let $ \tau \in [t, T] $.  For $ g(x) = \ln(x) +1 $, it is identical to show that $ \big\{ \ln(X_{t_n}) \big\}_{n=1}^\infty $ is uniformly bounded by an integrable random variable. We apply It\^{o}'s formula to $ \ln(X_\tau) $ and obtain
		\begin{align*}
			\ln(X_\tau) = \
			& \ln(x) +  \int_t^\tau \Big[ \frac{\sigma_s \pi_s \lambda_s}{X_s}
			- \frac{\sigma^2_s \pi^2_s}{2X_s^2} \Big] ds
			+ \int_t^\tau  \frac{\sigma_s \pi_s}{X_s} dW^{(1)}_s \\
			& + \int_t^\tau \int_\mathbb{R} \Big[ \ln \big( X_s + \sigma_s\pi_s\gamma_1(s,\zeta) \big) - \ln(X_s) - \frac{\sigma_s \pi_s}{X_s}\gamma_1(s,\zeta) \Big] \nu(d\zeta) ds \\
			& + \int_t^\tau \int_\mathbb{R} \Big[ \ln \big( X_s + \sigma_s\pi_s\gamma_1(s,\zeta) \big) - \ln(X_s) \Big] \tilde{N}(ds,d\zeta).
		\end{align*}
		With taking expectation, we choose some constant $ c_1 $ such that 
		\begin{align*}
			\mathbb{E} \big( \ln(X_\tau) \big) \leq\
			&  c_1 \Bigg[
			1 +  \mathbb{E} \bigg( \int_0^T \frac{\sigma_s^2 \pi_s^2 }{X_s^2} ds \bigg)
			+ \mathbb{E} \bigg( \bigg[ \sup_\tau \int_t^\tau  \frac{\sigma_s \pi_s}{X_s} dW^{(1)}_s \bigg]^2 \bigg) \\
			& \quad \
			+ \mathbb{E} \bigg( \int_0^T \int_\mathbb{R} \Big[ \ln \big( X_s + \sigma_s\pi_s\gamma_1(s,\zeta) \big) - \ln(X_s) \Big]^2 \nu(d\zeta) ds \bigg) \\
			& \quad \
			+ \mathbb{E} \bigg( \bigg[ \sup_\tau \int_t^\tau \int_\mathbb{R} \Big[ \ln \big( X_s + \sigma_s\pi_s\gamma_1(s,\zeta) \big) - \ln(X_s) \Big] \tilde{N}(ds,d\zeta) \bigg]^2 \bigg)
			\Bigg].
		\end{align*}
		By Doob's martingale maximal inequalities, we have
		\begin{align*}
			\mathbb{E} \big( \ln(X_\tau) \big) \leq \
			& c_1 \Bigg[
			1 +  \mathbb{E} \bigg( \int_0^T \frac{\sigma_s^2 \pi_s^2 }{X_s^2} ds \bigg)
			+ 4 \mathbb{E} \bigg( \bigg[  \int_0^T  \frac{\sigma_s \pi_s}{X_s} dW^{(1)}_s \bigg]^2 \bigg) \\
			& \quad \
			+ \mathbb{E} \bigg(  \int_0^T \int_\mathbb{R} \Big[ \ln \big( X_s + \sigma_s\pi_s\gamma_1(s,\zeta) \big) - \ln(X_s) \Big]^2 \nu(d\zeta) ds  \bigg) \\
			& \quad \
			+ 4 \mathbb{E} \bigg(  \bigg[  \int_0^T \int_\mathbb{R} \Big[ \ln \big( X_s + \sigma_s\pi_s\gamma_1(s,\zeta) \big) - \ln(X_s) \Big] \tilde{N}(ds,d\zeta) \bigg]^2 \bigg)
			\Bigg].
		\end{align*}
		Then by It\^{o} isometries, we have
		\begin{align*}
			\mathbb{E}\big( \ln(X_\tau) \big) \leq\
			& c_1 \Bigg[
			1 +  5\mathbb{E} \bigg( \int_0^T \frac{\sigma_s^2 \pi_s^2 }{X_s^2} ds \bigg) \\
			& \quad \
			+ 5\mathbb{E} \bigg( \int_0^T \int_\mathbb{R} \Big[ \ln \big( X_s + \sigma_s\pi_s\gamma_1(s,\zeta) \big) - \ln(X_s) \Big]^2 \nu(d\zeta) ds \bigg)
			\Bigg].
		\end{align*}
		Finally by Definition \ref{def_1}, we prove that $ \big\{ \ln(X_{t_n}) \big\}_{n=1}^\infty $ is uniformly bounded by an integrable random variable.
		
		For $ g(x) = x^{1-\alpha} + x^{1-\beta} $, it is identical to show that $ \{ X^{1-r}_{t_n} \}_{n=1}^\infty $, $ r > 0 $ and $ r \neq 1, $ is uniformly bounded by an integrable random variable. We apply It\^{o}'s formula to $ X^{1-r}_\tau $ and obtain
		\begin{align*}
			X^{1-r}_\tau = \
			& x^{1-r} + [1-r] \bigg[ \int_t^\tau \Big[ \frac{\sigma_s \pi_s \lambda_s}{X^{r}_s} 
			- \frac{r\sigma^2_s \pi^2_s}{2X^{r+1}_s} \Big] ds
			+ \int_t^\tau \frac{\sigma_s \pi_s}{X^{r}_s} dW^{(1)}_s \bigg] \\
			& + \int_t^\tau \int_\mathbb{R} \Big[ \big[ X_s + \sigma_s\pi_s\gamma_1(s,\zeta) \big]^{1-r} - X^{1-r}_s - [1-r]\frac{\sigma_s \pi_s}{X^{r}_s}\gamma_1(s,\zeta) \Big] \nu(d\zeta) ds \\
			& + \int_t^\tau \int_\mathbb{R} \Big[ \big[ X_s + \sigma_s\pi_s\gamma_1(s,\zeta) \big]^{1-r} - X^{1-r}_s \Big] \tilde{N}(ds,d\zeta).
		\end{align*}
		With taking expectation, we choose some constant $ c_2 $ such that 
		\begin{align*}
			\mathbb{E} (X^{1-r}_\tau) \leq\
			& c_2 \Bigg[
			1 +  \mathbb{E} \bigg( \int_0^T \frac{\sigma^2_s \pi^2_s}{X^{2r}_s} ds \bigg)
			+ \mathbb{E} \bigg( \bigg[ \sup_\tau \int_t^\tau  \frac{\sigma_s \pi_s}{X^{r}_s} dW^{(1)}_s \bigg]^2  \bigg) \\
			& \quad \
			+ \mathbb{E} \bigg(  \int_0^T \int_\mathbb{R} \Big[ \big[ X_s + \sigma_s\pi_s\gamma_1(s,\zeta) \big]^{1-r} - X^{1-r}_s \Big]^2 \nu(d\zeta) ds \bigg) \\
			& \quad \
			+ \mathbb{E} \bigg( \bigg[ \sup_\tau \int_t^\tau \int_\mathbb{R} \Big[ \big[ X_s + \sigma_s\pi_s\gamma_1(s,\zeta) \big]^{1-r} - X^{1-r}_s \Big] \tilde{N}(ds,d\zeta) \bigg]^2 \bigg)
			\Bigg].
		\end{align*}
		By Doob's martingale maximal inequalities, we have
		\begin{align*}
			\mathbb{E} (X^{1-r}_\tau) \leq\
			& c_2 \Bigg[
			1 +  \mathbb{E} \bigg( \int_0^T \frac{\sigma^2_s \pi^2_s}{X^{2r}_s} ds \bigg)
			+ 4 \mathbb{E} \bigg( \bigg[ \int_0^T  \frac{\sigma_s \pi_s}{X^{r}_s} dW^{(1)}_s \bigg]^2 \bigg)\\
			& \quad \
			+ \mathbb{E} \bigg( \int_0^T \int_\mathbb{R} \Big[ \big[ X_s + \sigma_s\pi_s\gamma_1(s,\zeta) \big]^{1-r} - X^{1-r}_s \Big]^2 \nu(d\zeta) ds \bigg)\\
			& \quad \
			+ 4 \mathbb{E} \bigg( \bigg[ \int_0^T \int_\mathbb{R} \Big[ \big[ X_s + \sigma_s\pi_s\gamma_1(s,\zeta) \big]^{1-r} - X^{1-r}_s \Big] \tilde{N}(ds,d\zeta) \bigg]^2 \bigg)
			\Bigg].
		\end{align*}
		Then by It\^{o} isometries, we have
		\begin{align*}
			\mathbb{E} (X^{1-r}_\tau) \leq\
			& c_2 \Bigg[
			1 + 5 \mathbb{E} \bigg( \int_0^T \frac{\sigma^2_s \pi^2_s}{X^{2r}_s} ds \bigg) \\
			& \quad \
			+ 5 \mathbb{E} \bigg( \int_0^T \int_\mathbb{R} \Big[ \big[ X_s + \sigma_s\pi_s\gamma_1(s,\zeta) \big]^{1-r} - X^{1-r}_s \Big]^2 \nu(d\zeta) ds \bigg)
			\Bigg].
		\end{align*}
		Finally by Definition \ref{def_1}, we prove that $ \{ X^{1-r}_{t_n} \}_{n=1}^\infty $ is uniformly bounded by an integrable random variable.
	\end{proof}
	
	\begin{lemma} \label{A.2}
		If the total wealth $ \hat{X}_t $ is defined by
		\begin{align*}
			d\hat{X}_t = \sigma(Y_t) \tilde{\pi}(t,\hat{X}_t,Y_t) \Big[ \lambda(Y_t)dt + dW_t^{(1)} + \int_\mathbb{R} \gamma_1(t,\zeta) \tilde{N}(dt,d\zeta) \Big], 
		\end{align*}
		where
		\begin{align*}
			\tilde{\pi}(t,x,y)
			= \frac{ -\lambda(y)\hat{U}_x(t,x,y) - \rho a(y) \hat{U}_{xy}(t,x,y)}
			{\sigma(y)\hat{U}_{xx}(t,x,y)}
			+ \frac{\hat{U}_x(t,x,y) }
			{\sigma(y)\hat{U}_{xx}(t,x,y)} \int_\mathbb{R}  \gamma_1(t, \zeta) \nu(d\zeta),
		\end{align*}
		with $ \hat{U}(t,x,y) $ given by \eqref{U_hat}, and $ \hat{X}_t = x, Y_t = y $. Then the optimal portfolio $ \tilde{\pi}(t,x,y)  $ is admissible under Assumption \ref{Assume_UT}.
	\end{lemma}
	\begin{proof}
		We denote $ \sigma_t = \sigma(Y_t) $ and $ \tilde{\pi}_t = \tilde{\pi}(t, \hat{X}_t, Y_t) $.
		Because $ \tilde{\pi}_t $ is continuous with respect to $ t $, therefore it is progressively measurable. 
		
		We apply \eqref{sim7}, \eqref{sim13} and \eqref{sim14} into equation \eqref{U_hat}, and obtain $ \hat{U}_x(t,x,y) \sim U_{_T}'(x) $, $ \hat{U}_{xy}(t,x,y) \sim U_{_T}'(x) $ and $ \hat{U}_{xx}(t,x,y) \sim U_{_T}''(x) $, which imply $ \tilde{\pi}_t \sim U_{_T}'(x)/U_{_T}''(x) $.
		Under Case 1 of Assumption \ref{Assume_UT}, we have
		\begin{align*}
			\lim_{x \rightarrow \infty} \frac{U_{_T}'(x)}{xU_{_T}''(x)}
			= \lim_{x \rightarrow \infty} \frac{x^{-1}}{x[-x^{-2}]} = -1.
		\end{align*}
		Under Case 2 of Assumption \ref{Assume_UT}, we have
		\begin{align*}
			\lim_{x \rightarrow \infty} \frac{U_{_T}'(x)}{xU_{_T}''(x)}
			= \lim_{x \rightarrow \infty} \frac{c_1x^{-\alpha}+c_2x^{-\beta}}{x[-c_1\alpha x^{-\alpha-1} - c_2\beta x^{-\beta-1}]} 
			= -\frac{1}{\alpha} \text{ or } -\frac{1}{\beta}.
		\end{align*}
		Thus we find $ \big| \sigma(y) \tilde{\pi}_t/x \big| \leq c $ with some constant $ c $. By definition of  $ \hat{X}_t $, $ \tilde{\pi}_t $ yields that $ \hat{X}_t $ is strictly positive.
		
		The proof of Lemma A.1 in \cite{Kumar} had already given that
		\begin{align*}
			\mathbb{E} \bigg( \int_0^T \frac{\sigma^2_t \tilde{\pi}^2_t}{\hat{X}^{2r}_t} dt \bigg) < \infty, \quad r>0.
		\end{align*}
		We find in Theorem \ref{thm} that $ \hat{U}(t,x,y) $ is the approach of $ U(t,x,y) $. The function $ U(t,x,y) $ satisfies the growth conditions \eqref{growth}. So does the $ \hat{U}(t,x,y) $. Thus we have
		\begin{align*}
			\mathbb{E} \bigg( \int_0^T \int_\mathbb{R} \Big[ \ln \big( \hat{X}_t + \sigma_t\tilde{\pi}_t\gamma_1(t,\zeta)\big) - \ln(\hat{X}_t) \Big]^2 \nu(d\zeta) dt \bigg) < \infty,
		\end{align*}
		under Case 1 of Assumption \ref{Assume_UT}, and 
		\begin{align*}
			\mathbb{E} \bigg( \int_0^T \int_\mathbb{R} \Big[ \big[ \hat{X}_t + \sigma_t\tilde{\pi}_t\gamma_1(t,\zeta) \big]^{1-r} - \hat{X}^{1-r}_t \Big]^2 \nu(d\zeta) dt \bigg) < \infty, \quad r>0 \text{ and } r \neq 1,
		\end{align*}
		under Case 2 of Assumption \ref{Assume_UT}.
		By Definition \ref{def_1}, we prove that the optimal portfolio $ \tilde{\pi}(t,x,y)  $ is admissible under Assumption \ref{Assume_UT}.
	\end{proof}

\end{document}